\documentclass[a4paper11pt]{article}
\usepackage{amssymb,amsthm, latexsym}

\newcommand{\words}{\{0,1\}^*}
\newcommand{\cantor}{\{0,1\}^\mathbb{N}}
\newcommand{\mix}{\{0,1\}^{\leq\omega}}
\newcommand{\segment}{\!\upharpoonright \!}
\theoremstyle{plain}
\newtheorem{theorem}{Theorem}[section]
\newtheorem{proposition}[theorem]{Proposition}
\newtheorem{corollary}[theorem]{Corollary}
\newtheorem{definition}[theorem]{Definition}
\newtheorem*{theorem*}{Theorem}
\newtheorem{example}[theorem]{Example}
\theoremstyle{remark}
\newtheorem{remark}[theorem]{{\bf Remark.}}
\title{{\bf\huge Kolmogorov Complexity in perspective}
\bigskip\bigskip
\\{\bf Part I: Information Theory and Randomness}\footnote{%
Published in Synthese, 2008.
                                 }%
}
\author{Marie Ferbus-Zanda\\
{\footnotesize LIAFA, CNRS \& Universit\'e Paris 7}\\
{\footnotesize case 7014}\\
{\footnotesize 75205 Paris Cedex 13 France}\\
{\footnotesize Marie.Ferbus@liafa.jussieu.fr}
\and Serge Grigorieff\\
{\footnotesize LIAFA, CNRS \& Universit\'e Paris 7}\\
{\footnotesize case 7014}\\
{\footnotesize 75205 Paris Cedex 13 France}\\
{\footnotesize seg@liafa.jussieu.fr}}
\date{}
\begin{document}
\maketitle
%%%%%%%%%%%%%%%%%%%%%%%%%%%%%%%%%%%%%%
%%%%%%%%%%%%%%%%%%%%%%%%%%%%%%%%%%%%%%
%%%%%%%%%%%%%%%%%%%%%%%%%%%%%%%%%%%%%%
%%%%%%%%%%%%%%%%%%%%%%%%%%%%%%%%%%%%%%
%%%%%%%%%%%%%%%%%%%%%%%%%%%%%%%%%%%%%%
%%%%%%%%%%%%%%%%%%%%%%%%%%%%%%%%%%%%%%
%%%%%%%%%%%%%%%%%%%%%%%%%%%%%%%%%%%%%%
%%%%%%%%%%%%%%%%%%%%%%%%%%%%%%%%%%%%%%
%
%%%%%%%%%%%%%%%%%%%%%%%%%%%%%%%%%%%%%%
%%%%%%%%%%%%%%%%%%%%%%%%%%%%%%%%%%%%%%
%%%%%%%%%%%%%%%%%%%%%%%%%%%%%%%%%%%%%%
%%%%%%%%%%%%%%%%%%%%%%%%%%%%%%%%%%%%%%
% \renewcommand\abstractname{Abstract for Parts I and II}

\begin{abstract}
%	
%%%%%%%%%%%%%%%%%%%%%%%%%%%%%%%%%%%%%%
%   
\noindent                                                               
We survey diverse approaches to the notion of information:
from Shannon entropy to Kolmogorov complexity.
Two of the main applications of Kolmogorov complexity are presented:
randomness and classification.
The survey is divided in two parts in the same volume.
\medskip
\\
Part I is dedicated to information 
theory and the mathematical
formalization of randomness based on Kolmogorov complexity.
This last application goes back to the 60's and 70's with the work
of Martin-L\"of, Schnorr, Chaitin, Levin,
and has gained new impetus in the last years.
%
%%%%%%%%%%%%%%%%%%%%%%%%%%%%%%%%%%%%%
%%%%%%%%%%%%%%%%%%%%%%%%%%%%%%%%%%%%%
%
\bigskip
\\
\noindent{\bf  Keywords:} Logic, Computer Science, Algoritmmic 
Information Theory, Shannon Information Theory, Kolmogorov Complexity, 
Randomness.
%	
%%%%%%%%%%%%%%%%%%%%%%%%%%%%%%%%%%%%%%%%%%%%%%%%%%%%%%%%%%%%%%%%%%%   
\end{abstract}
%
%%%%%%%%%%%%%%%%%%%%%%%%%%%%%%%%%%%%%%
%%%%%%%%%%%%%%%%%%%%%%%%%%%%%%%%%%%%%%
%
{\footnotesize\tableofcontents}
\normalsize
%%%%%%%%%%%%%%%%%%%%%%%%%%%%%%%%%%%%%%
%%%%%%%%%%%%%%%%%%%%%%%%%%%%%%%%%%%%%%
%
%
%%%%%%%%%%%%%%%%%%%%%%%%%%%%%%%%%%%%%%
%%%%%%%%%%%%%%%%%%%%%%%%%%%%%%%%%%%%%%
%%%%%%%%%%%%%%%%%%%%%%%%%%%%%%%%%%%%%%
%%%%%%%%%%%%%%%%%%%%%%%%%%%%%%%%%%%%%%
\newpage
%
%%%%%%%%%%%%%%%%%%%%%%%%%%%%%%%%%%%%%%
%%%%%%%%%%%%%%%%%%%%%%%%%%%%%%%%%%%%%%
%%%%%%%%%%%%%%%%%%%%%%%%%%%%%%%%%%%%%%
%%%%%%%%%%%%%%%%%%%%%%%%%%%%%%%%%%%%%%
\noindent {\bf Note.}
Following Robert Soare's recommendations (\cite{soareBSL}, 1996),
which have now gained large agreement, we write
{\em computable} and {\em computably enumerable}
in place of the old fashioned
{\em recursive} and {\em recursively enumerable}.
\medskip
\\\noindent {\bf Notation.}
By $\log x$ (resp. $\log_s x$) we mean the logarithm of $x$
in base $2$ (resp. base $s$ where $s\geq2$).
The ``floor" and ``ceil" of a real number $x$ are denoted by
$\lfloor x\rfloor$ and $\lceil x\rceil$:
they are respectively the largest integer $\leq x$
and the smallest integer $\geq x$.
Recall that, for $s\geq2$,
the length of the base $s$ representation of an integer $k$
is $\ell\geq1$ if and only if $s^{\ell-1} \leq k < s^\ell$.
Thus, the length of the base $s$ representation of an integer $k$
is
$1+\lfloor\log_s k\rfloor=1+\lfloor\frac{\log k}{\log s}\rfloor$.
\\
The number of elements of a finite family $\+F$ is denoted
by $\sharp\+F$.
\\
The length of a word $u$ is denoted by $|u|$.
%
%%%%%%%%%%%%%%%%%%%%%%%%%%%%%%%%%%%
%%%%%%%%%%%%%%%%%%%%%%%%%%%%%%%%%%%
%%%%%%%%%%%%%%%%%%%%%%%%%%%%%%%%%%%
\section{Three approaches to a quantitative definition
of information}
%%%%%%%%%%%%%%%%%%%%%%%%%%%%%%%%%%%
%%%%%%%%%%%%%%%%%%%%%%%%%%%%%%%%%%%
%%%%%%%%%%%%%%%%%%%%%%%%%%%%%%%%%%%
%
A title borrowed from Kolmogorov' seminal paper (\cite{kolmo65}, 1965).
%
%
%%%%%%%%%%%%%%%%%%%%%%%%%%%%%%%%%%%
\subsection{Which information?}
%%%%%%%%%%%%%%%%%%%%%%%%%%%%%%%%%%%
%
%------------------------------------------------------------------
\subsubsection{About anything...}
%------------------------------------------------------------------
About anything can be seen as conveying information.
As usual in mathematical modelization,
we retain only a few features of some real entity or process,
and associate to them some finite or infinite mathematical objects.
For instance,
\begin{itemize}
\item[\textbullet]
- an integer or a rational number or a word in some alphabet,
\\- a finite sequence or a finite set of such objects,
\\- a finite graph,...
\item[\textbullet]
- a real,
\\- a finite or infinite sequence of reals or a set of reals,
\\- a function over words or numbers,...
\end{itemize}
This is very much as with probability spaces. For instance,
to modelize the distributions of $6$ balls into $3$ cells,
(cf. Feller, ~\cite{feller}, \S I.2, II.5)
we forget everything about the nature of balls and cells
and of the distribution process,
retaining only two questions:
``how many balls in each cell?" and
``are the balls and cells distinguishable or not?".
Accordingly, the modelization considers
\\
- either the $729=3^6$ maps from the set of balls
into the set of cells
in case the balls are distinguishable and so are the cells
(this is what is done in Maxwell-Boltzman statistics),
\\
- or the $28=\left(\begin{array}{c} 6+(3-1) \\6 \end{array}\right)$ 
triples of non negative integers
with sum\footnote{This value is easily obtained by identifying such
a triple with a binary word
with six letters $0$ for the six balls
and two letters $1$ to mark the partition in the three cells.}
 $6$
in case the cells are distinguishable but not the balls
(this is what is done in Bose-Einstein statistics)
\\
- or the $7$ sets of at most $3$ integers with sum $6$
in case the balls are undistinguishable and so are the cells.
%
%
%------------------------------------------------------------------
\subsubsection{Especially words}\label{sss:words}
%------------------------------------------------------------------
In information theory, special emphasis is made on information
conveyed by words on finite alphabets.
I.e., on {\em sequential information} as opposed to the obviously
massively parallel and interactive distribution of information
in real entities and processes.
A drastic reduction which allows for mathematical developments
(but also illustrates the Italian saying
``traduttore, traditore!").
\medskip\\
As is largely popularized by computer science, any finite alphabet
with more than two letters can be reduced to one with exactly
two letters.
For instance, as exemplified by the ASCII code
(American Standard Code for Information Interchange),
any symbol used in written English
-- namely the lowercase and uppercase letters, the decimal digits,
the diverse punctuation marks, the space, apostrophe, quote,
left and right parentheses --
together with some simple typographical commands
-- such as tabulation, line feed, carriage return or
``end of file" --
can be coded by binary words of length $7$
(corresponding to the $128$ ASCII codes).
This leads to a simple way to code any English text by
a binary word (which is $7$ times longer)\footnote{For
other European languages which have a lot of
diacritic marks, one has to consider the $256$
codes of Extended ASCII which have length $8$.
And for non European languages, one has to turn to the
$65~536$ codes of Unicode which have length $16$.}.
\medskip\\
Though quite rough, the length of a word is the basic measure
of its information content.
Now, a fairness issue faces us:
richer the alphabet, shorter the word.
Considering groups of $k$ successive letters as new letters
of a super-alphabet, one trivially divides the length by $k$.
For instance, a length $n$ binary word becomes a length
$\lceil \frac{n}{256}\rceil$ word with the usual packing of bits
by groups of $8$ (called bytes) which is done in computers.
\\
This is why all considerations about the length of words
will always be developed relative to binary alphabets.
A choice to be considered as a {\em normalization of length}.
\medskip\\
Finally, we come to the basic idea to measure the information
content of a mathematical object $x$ :
\medskip\\\medskip\centerline{\em
\begin{tabular}{|rcl|}
\hline
information content of $x$
&=&
\begin{tabular}{l}
length of a shortest binary word\smallskip\\
which ``encodes" $x$
\end{tabular}
\\\hline
\end{tabular}}
What do we mean precisely by ``encodes" is the crucial question.
Following the trichotomy pointed by Kolmogorov in
\cite{kolmo65}, 1965, we survey three approaches.
%
%
%%%%%%%%%%%%%%%%%%%%%%%%%%%%%%%%%%%
\subsection{Combinatorial approach: entropy}\label{s:combinatorial}
%%%%%%%%%%%%%%%%%%%%%%%%%%%%%%%%%%%
%
%------------------------------------------------------------------
\subsubsection{Constant-length codes}\label{sss:constantlength}
%------------------------------------------------------------------
Let us consider the family $A^n$ of length $n$ words in an alphabet
$A$ with $s$ letters $a_1,...,a_s$.
Coding the $a_i$'s by binary words $w_i$'s all of length
$\lceil\log s\rceil$, to any word $u$ in $A^n$ we can associate
the binary word $\xi$ obtained by substituting the $w_i$'s to the
occurrences of the $a_i$'s in $u$.
Clearly, $\xi$ has length $n\lceil\log s\rceil$.
Also, the map $u\mapsto\xi$ from the set $A^*$ of words in alphabet
$A$ to the set $\{0,1\}^*$ of binary words is very simple. Mathematically, considering on $A^*$ and $\{0,1\}^*$
the algebraic structure of monoid
given by the concatenation product of words,
this map $u\mapsto\xi$ is a morphism
since the image of a concatenation $uv$ is the concatenation
of the images of $u$ and $v$.
%
%
%------------------------------------------------------------------
\subsubsection{Variable-length prefix codes}\label{sss:varlength}
%------------------------------------------------------------------
Instead of coding the $s$ letters of $A$ by binary words of length
$\lceil\log s\rceil$, one can code the $a_i$'s by binary words $w_i$'s
having different lengthes so as to associate
short codes to most frequent letters and long codes to rare ones.
This is the basic idea of compression.
Using such codes, the substitution of the $w_i$'s to the
occurrences of the $a_i$'s in a word $u$ gives a binary word $\xi$.
And the map $u\mapsto\xi$ is again very simple.
It is still a morphism from the monoid of words on alphabet $A$
to the monoid of binary words
and can also be computed by a finite automaton.
\medskip
\\
Now, we face a problem: can we recover $u$ from $\xi$ ?
i.e., is the map $u\mapsto \xi$ injective?
In general the answer is no.
However, a simple sufficient condition to ensure decoding
is that the family $w_1,...,w_s$ be a so-called
{\em prefix-free code} (or {\em prefix code}).
Which means that if $i\neq j$ then $w_i$ is not a prefix of $w_j$.
\begin{quote}
This condition insures that there is a unique $w_{i_1}$
which is a prefix of $\xi$. 
Then, considering the associated suffix $\xi_1$ of $v$
(i.e., $v=w_{i_1}\xi_1$) there is a unique $w_{i_2}$
which is a prefix of $\xi_1$, i.e., $u$ is of the form
$u=w_{i_1}w_{i_2}\xi_2$.
And so on.
\end{quote}
\noindent
Suppose the numbers of occurrences in $u$ of the letters
$a_1,...,a_s$ are $m_1,...,m_s$,
so that the length of $u$ is $n=m_1+...+m_s$.
Using a prefix-free code $w_1,...,w_s$, the binary word $\xi$
associated to $u$ has length $m_1|w_1|+...+m_s|w_s|$.
A natural question is, given $m_1,...,m_s$,
{\em how to choose the prefix-free code $w_1,...,w_s$
so as to minimize the length of $\xi$~?}
\\
Huffman (\cite{huffman}, 1952) found a very efficient
algorithm (which has linear time complexity if the frequencies
are already ordered).
This algorithm (suitably modified to keep its top efficiency
for words containing long runs of the same data)
is nowadays used in nearly every application
that involves the compression and transmission of data:
fax machines, modems, networks,...
%
%------------------------------------------------------------------
\subsubsection{Entropy of a distribution of frequencies}
\label{sss:entropy}
%------------------------------------------------------------------
The intuition of the notion of entropy in information theory
is as follows.
Given natural integers $m_1,...,m_s$, consider the family
$\+F_{m_1,...,m_s}$ of length $n=m_1+...+m_s$ words
of the alphabet $A$ in which there are exactly $m_1,...,m_s$
occurrences of letters $a_1,...,a_s$.
How many binary digits are there in the binary representation
of the number of words in $\+F_{m_1,...,m_s}$~?
It happens (cf. Proposition \ref{p:H})
that this number is essentially linear in
$n$, the coefficient of $n$ depending solely on the frequencies
$\frac{m_1}{n},...,\frac{m_s}{n}$.
It is this coefficient which is called the entropy $H$ of the
distribution of the frequencies $\frac{m_1}{n},...,\frac{m_s}{n}$.
\begin{definition}[Shannon, \cite{shannon48}, 1948]
Let $f_1,...,f_s$ be a distribution of frequencies,
i.e., a sequence of reals in $[0,1]$ such that $f_1+...+f_s=1$.
The entropy of $f_1,...,f_s$ is the real
$$
H = -(f_1\log(f_1) +...+f_s\log(f_s))
$$
\end{definition}
\begin{proposition}[Shannon, \cite{shannon48}, 1948]\label{p:H}
Let $m_1,...,m_s$ be natural integers and $n=m_1+...+m_s$.
Then, letting $H$ be the entropy of the distribution of frequencies
$\frac{m_1}{n},...,\frac{m_s}{n}$,
the number $\sharp\+F_{m_1,...,m_s}$ of words in
$\+F_{m_1,...,m_s}$ satisfies
$$
\log(\sharp\+F_{m_1,...,m_s}) = nH + O(\log n)
$$
where the bound in $O(\log n)$ depends solely on $s$
and not on $m_1,...,m_s$.
\end{proposition}
\begin{quote}
\begin{proof}
The set $\+F_{m_1,...,m_s}$ contains
$\frac{n!}{m_1!\times...\times m_s!}$ words.
Using Stirling's approximation of the factorial function
(cf.~\cite{feller}), namely
$x! = \sqrt{2\pi} \ x^{x+\frac{1}{2}} \ e^{-x+\frac{\theta}{12}}$
where $0<\theta<1$, and equality $n=m_1+...+m_S$, we get
\begin{eqnarray*}
\log(\frac{n!}{m_1!\times...\times m_s!})
&=& (\sum_i m_i)\log(n) - (\sum_i m_i\log m_i)
\\&& 
+\frac{1}{2} \log(\frac{n}{m_1\times...\times m_s})
-(s-1)\log\sqrt{2\pi} + \alpha
\end{eqnarray*}
where $|\alpha|\leq\frac{s}{12}\log e$.
The difference of the first two terms is equal to
$n[\sum_i\frac{m_i}{n}\log(\frac{m_i}{n})] = nH$
and the remaining sum is $O(\log n)$ since
$n^{1-s} \leq \frac{n}{m_1\times...\times m_s} \leq n$.
\end{proof}
\end{quote}
{\em $H$ has a striking significance in terms of information
content and compression.}
Any word $u$ in $\+F_{m_1,...,m_s}$ is uniquely
characterized by its rank in this family
(say relatively to the lexicographic ordering on words
in alphabet $A$).
In particular, the binary representation of this rank ``encodes"
$u$.
Since this rank is $<\sharp\+F_{m_1,...,m_s}$,
its binary representation has length $\leq nH$ up to
an $O(\log n)$ term.
Thus, $nH$ can be seen as an upper bound of the information
content of $u$.
Otherwise said, the $n$ letters of $u$ are encoded by $nH$
binary digits.
In terms of compression
(nowadays so popular with the zip-like softwares),
{\em $u$ can be compressed to $nH$ bits},
i.e.,
{\em the mean information content
(which can be seen as the compression size in bits)
of a letter of $u$ is $H$}.
\\
Let us look at two extreme cases.
\\
\textbullet\
If all frequencies $f_i$ are equal to $\frac{1}{s}$
then the entropy is $\log(s)$, so that the mean information content
of a letter of $u$ is $\log(s)$,
i.e., there is no better  (prefix-free) coding than that described in
\S\ref{sss:constantlength}.
\\
\textbullet\
In case some of the frequencies is $1$ (hence all other ones being $0$),
the information content of $u$ is reduced to its length $n$,
which, written in binary, requires $\log(n)$ bits.
As for the entropy, it is $0$
(with the usual convention $0\log 0=0$, justified by the fact that
$\lim_{x\to0}x\log x=0$).
The discrepancy between $nH=0$ and the true information content
$\log n$ comes from the $O(\log n)$ term in PropositionÊ\ref{p:H}.
%
%------------------------------------------------------------------
\subsubsection{Shannon's source coding theorem for symbol codes}
\label{sss:sourcecoding}
%------------------------------------------------------------------
The significance of the entropy explained above has been given
a remarkable and precise form by
Claude Elwood Shannon (1916-2001) in his celebrated paper
\cite{shannon48}, 1948.
It's about the length of the binary word $\xi$ associated
to $u$ via a prefix-free code.
Shannon proved
\\- a lower bound of $|\xi|$ valid whatever be
the prefix-free code $w_1,...,w_s$,
\\- an upper bound, quite close to the lower bound,
valid for particular prefix-free codes $w_1,...,w_s$
(those making $\xi$ shortest possible, for instance those given
by Huffman's algorithm).
\begin{theorem}[Shannon, \cite{shannon48}, 1948]\label{thm:H}
Suppose the numbers of occurrences in $u$ of the letters
$a_1,...,a_s$ are $m_1,...,m_s$. Let $n=m_1+...+m_s$.
\medskip\\
1. For every prefix-free sequence of binary words $w_1,...,w_s$
(which are to code the letters $a_1,...,a_s$),
the binary word $\xi$ obtained by substituting $w_i$ to each
occurrence of $a_i$ in $u$ satisfies
$$
nH \leq |\xi|
$$
where
$H\ =\ -(\frac{m_1}{n}\log(\frac{m_1}{n}) +...
+ \frac{m_s}{n}\log(\frac{m_s}{n}))$
is the entropy of the considered distribution of frequencies
$\frac{m_1}{n},...,\frac{m_s}{n}$.
\medskip\\
2. There exists a prefix-free sequence of binary words $w_1,...,w_s$
such that
$$
nH \leq |\xi| < n(H+1)
$$
\end{theorem}
\begin{quote}
\begin{proof}
First, we recall two classical results.
\medskip\\
{\em Kraft's inequality.}
Let $\ell_1,...,\ell_s$ be a finite sequence of integers.
Inequality $2^{-\ell_1}+...+2^{-\ell_s} \leq 1$ holds
if and only if there exists a prefix-free sequence of binary words
$w_1,...,w_s$ such that $\ell_1=|w_1|,...,\ell_s=|w_s|$.
\medskip\\
{\em Gibbs' inequality.}
Let $p_1,...,p_s$ and $q_1,...,q_s$ be two probability
distributions, i.e., the $p_i$'s (resp. $q_i$'s) are in $[0,1]$
and have sum $1$.
Then
$-\sum p_i\log(p_i) \leq -\sum p_i\log(q_i)$
with equality if and only if $p_i=q_i$ for all $i$.
\medskip\\
{\em Proof of Point 1 of Theorem \ref{thm:H}}.
Set $p_i=\frac{m_i}{n}$ and $q_i=\frac{2^{-|w_i|}}{S}$
where $S=\sum_i 2^{-|w_i|}$. Then
\medskip\\
$|\xi| = \sum_i m_i|w_i|
=n[\sum_i\frac{m_i}{n}(-\log(q_i) -\log S)]$

\hfill{$\geq n[-(\sum_i\frac{m_i}{n}\log(\frac{m_i}{n}) -\log S]
=n[H-\log S]
\geq nH$}
\medskip\\
The first inequality is an instance of Gibbs' inequality.
For the last one, observe that $S\leq1$.
\medskip\\
{\em Proof of Point 2 of Theorem \ref{thm:H}}.
Set $\ell_i=\lceil-\log(\frac{m_i}{n})\rceil$.
Observe that $2^{-\ell_i} \leq \frac{m_i}{n}$.
Thus, $2^{-\ell_1}+...+2^{-\ell_s} \leq 1$.
Applying Kraft inequality, we see that there exists a prefix-free
family of words $w_1,...,w_s$ with lengthes $\ell_1,...,\ell_s$.
\\
We consider the binary word $\xi$ obtained via this prefix-free
code, i.e., $\xi$ is obtained by substituting $w_i$ to each
occurrence of $a_i$ in $u$.
Observe that
$-\log(\frac{m_i}{n}) \leq \ell_i < -\log(\frac{m_i}{n})+1$.
Summing, we get $nH \leq |\xi| < n(H+1)$.
\end{proof}
\end{quote}
\noindent
In particular cases, the lower bound $nH$ can be achieved.
\begin{theorem}
In case the frequencies $\frac{m_i}{n}$'s are all negative
powers of two (i.e., $\frac{1}{2},\frac{1}{4},\frac{1}{8}$,...)
then the optimal $\xi$ (given by Huffman's algorithm)
satisfies $\xi=nH$.
\end{theorem}
%
%------------------------------------------------------------------
\subsubsection{Closer to the entropy}\label{sss:closer}
%------------------------------------------------------------------
In \S\ref{sss:entropy} and \ref{sss:sourcecoding}, we supposed
the frequencies to be known and did not consider the information
content of these frequencies.
We now deal with that question.
\\
Let us go back to the encoding mentioned at the start of
\S\ref{sss:entropy}.
A word $u$ in the family $\+F_{m_1,...,m_s}$
(of length $n$ words with exactly $m_1,...,m_s$ occurrences
of $a_1,...,a_s$)
can be recovered from the following data:
\\- the values of $m_1,...,m_s$,
\\- the rank of $u$ in $\+F_{m_1,...,m_s}$
    (relative to the lexicographic order on words).
\\
We have seen (cf. Proposition \ref{p:H})
that the rank of $u$ has a binary representation $\rho$
of length $\leq nH+O(\log n)$.
The integers $m_1,...,m_s$ are encoded by their binary
representations $\mu_1,...,\mu_s$ which all have length
$\leq 1+\lfloor\log n\rfloor$.
Now, to encode $m_1,...,m_s$ and the rank of $u$, we cannot
just concatenate $\mu_1,...,\mu_s,\rho$ : how would we know
where $\mu_1$ stops, where $\mu_2$ starts,...,
in the word obtained by concatenation?
Several tricks are possible to overcome the problem,
they are described in \S\ref{sss:codemany}.
Using Proposition \ref{p:code}, we set
$\xi = \langle \mu_1,...,\mu_s,\rho \rangle$
which has length
$|\xi| = |\rho|+O(|\mu_1|+...+|\mu_s|) = nH +O(\log n)$
(Proposition \ref{p:code} gives a much better bound but this is
of no use here).
Then $u$ can be recovered from $\xi$ which is a binary word
of length $nH+O(\log n)$.
Thus, asymptotically, we get a better upper bound than $n(H+1)$,
the one given by Shannon for prefix-free codes
(cf. Theorem \ref{thm:H}).
\\
Of course, $\xi$ is no more obtained from $u$ via a morphism
(i.e., a map which preserves concatenation of words)
between the monoid of words in alphabet $A$ and that of binary words.
\\
Notice that this also shows that prefix-free codes
are not the only way to efficiently encode into a binary word $\xi$
a word $u$ from alphabet $a_1,...,a_s$ for which the numbers
$m_1,...,m_s$ of occurrences of the $a_i$'s are known.

%
%------------------------------------------------------------------
\subsubsection{Coding finitely many words with one word}
\label{sss:codemany}
%------------------------------------------------------------------
How can we code two words $u,v$ with only one word?
The simplest way is to consider $u\$ v$ where $\$ $
is a fresh symbol outside the alphabet of $u$ and $v$.
But what if we want to stick to binary words?
As said above, the concatenation of $u$ and $v$ does not
do the job: how can one recover the prefix $u$ in $uv$?
A simple trick is to also concatenate the length of $|u|$ in
unary and delimitate it by a zero.
Indeed, denoting by $1^p$ the word
$\stackrel{p~\mbox{\footnotesize{times}}}{\overbrace{1\ldots1}}$,
one can recover $u$ and $v$ from the word $1^{|u|}0uv$ :
the length of the first block of $1$'s tells where to stop in
the suffix $uv$ to get $u$.
In other words, the map $(u,v)\to 1^{|u|}0uv$ is injective
from $\{0,1\}^*\times\{0,1\}^*\to\{0,1\}^*$.
In this way, the code of the pair $(u,v)$ has length $2|u|+|v|+1$.
This can obviously be extended to more arguments using the
map
$(u_1,...,u_s,v)\mapsto
1^{|u_1|}0^{|u_2|}\ldots\varepsilon^{|u_s|}
\varepsilon' u_1\ldots u_s v$
(where $\varepsilon=0$ is $s$ is even
and $\varepsilon=1$ is $s$ is odd
and $\varepsilon'=1-\varepsilon$.
\begin{proposition}\label{p:code}
Let $s\geq1$.
There exists a map
$\langle\ \rangle : (\{0,1\}^*)^{s+1}\to\{0,1\}^*$
which is injective and computable and such that,
for all $u_1,...,u_s,v\in\{0,1\}^*$,
$|\langle u_1,...,u_s,v \rangle|
= 2(|u_1|+...+|u_s|)+|v|+ 1$.
\end{proposition}
The following technical improvement will be needed
in Part II \S2.1.
\begin{proposition}\label{p:codeloglog}
There exists a map
$\langle\ \rangle : (\{0,1\}^*)^{s+1}\to\{0,1\}^*$
which is injective and computable and such that,
for all $u_1,...,u_s,v\in\{0,1\}^*$,
\begin{eqnarray*}
|\langle u_1,...,u_s,v \rangle|
&=& (|u_1|+...+|u_s|)
+ (\log|u_1|+...+\log|u_s|)
\\&&
+ 2(\log\log|u_1|+...+\log\log|u_s|) +|v| +O(1)
\end{eqnarray*}
\end{proposition}

\begin{quote}
\begin{proof}
We consider the case $s=1$, i.e., we want to code a pair $(u,v)$.
Instead of putting the prefix $1^{|u|}0$, let us put the binary
representation $\beta(|u|)$ of the number $|u|$ prefixed by
its length.
This gives the more complex code:
$1^{|\beta(|u|)|}0\beta(|u|)uv$
with length
$$
|u|+|v|+2(\lfloor\log|u|\rfloor+1)+1
\leq |u|+|v|+2\log|u| + 3
$$
The first block of ones gives the length of $\beta(|u|)$.
Using this length, we can get $\beta(|u|)$ as the factor following
this first block of ones.
Now, $\beta(|u|)$ is the binary representation of $|u|$, so we get
$|u|$ and can now separate $u$ and $v$ in the suffix $uv$.
\end{proof}
\end{quote}
%
%
%%%%%%%%%%%%%%%%%%%%%%%%%%%%%%%%%%%
\subsection{Probabilistic approach: ergodicity and lossy coding}
%%%%%%%%%%%%%%%%%%%%%%%%%%%%%%%%%%%
%
The abstract probabilistic approach allows for considerable
extensions of the results described in \S\ref{s:combinatorial}.
\medskip\\
First, the restriction to fixed given frequencies can be relaxed.
The probability of writing $a_i$ may depend on what has
already been written. For instance, Shannon's source coding theorem
has been extended to the so called
``ergodic asymptotically mean stationary source models".
\medskip\\
Second, one can consider a lossy coding: some length $n$ words in
alphabet $A$ are ill-treated or ignored.
Let $\delta$ be the probability of this set of words.
Shannon's theorem extends as follows:
\\- whatever close to $1$ is $\delta<1$,
one can compress $u$ only down to $nH$ bits.
\\- whatever close to $0$ is $\delta>0$,
one can achieve compression of $u$ down to $nH$ bits.
%
%
%
%%%%%%%%%%%%%%%%%%%%%%%%%%%%%%%%%%%
\subsection{Algorithmic approach: Kolmogorov complexity}
%%%%%%%%%%%%%%%%%%%%%%%%%%%%%%%%%%%
%%
%
%------------------------------------------------------------------
\subsubsection{Berry's paradox}\label{sss:berry}
%------------------------------------------------------------------
So far, we considered two kinds of binary codings for a word
$u$ in alphabet $a_1,...,a_s$.
The simplest one uses variable-length prefix-free codes
(\S\ref{sss:varlength}).
The other one codes the rank of $u$ as a member of some set
(\S\ref{sss:closer}).
\\
Clearly, there are plenty of other ways to encode any
mathematical object.
Why not consider all of them? And define the information
content of a mathematical object $x$ as
{\em the shortest univoque description of $x$ (written as
a binary word)}.
Though quite appealing, this notion is ill defined
as stressed by Berry's paradox\footnote{
Berry's paradox is mentioned by Bertrand Russell in 1908
(\cite{russell}, p.222 or 150), who credited G.G. Berry,
an Oxford librarian, for the suggestion.}:
\begin{quote}
Let $N$ be the
{\em lexicographically least binary word which cannot be
univoquely described by any binary word of length less
than $1000$}.
\end{quote}
This description of $N$ contains $106$ symbols of written
English (including spaces) and, using ASCII codes,
can be written as a binary word of length $106\times7=742$.
Assuming such a description to be well defined
would lead to a univoque description of $N$ in $742 $ bits,
hence less than $1000$,
a contradiction to the definition of $N$.
\\
The solution to this inconsistency is clear:
the quite vague notion of univoque description entering
Berry's paradox is used both inside the sentence
describing $N$ and inside the argument to get the
contradiction. A clash between two levels:
\\\indent{\textbullet} the would be formal 
level carrying the description
           of $N$
\\\indent{\textbullet}  and the meta level which 
carries the inconsistency argument.
\\
Any formalization of the notion of description
should drastically reduce its scope and totally forbid
any clash such as the above one.
%
%------------------------------------------------------------------
\subsubsection{The turn to computability}\label{sss:turn}
%------------------------------------------------------------------
To get around the stumbling block of Berry's paradox
and have a formal notion of description with wide scope,
Andrei Nikolaievitch Kolmogorov (1903--1987) made an
ingenious move: he turned to computability and
replaced {\em description} by {\em computation program}.
Exploiting the successful formalization of this a priori
vague notion which was achieved in the thirties\footnote{
Through the works of Alonzo Church (via lambda calculus),
Alan Mathison Turing (via Turing machines)
and Kurt G\"odel and Jacques Herbrand (via Herbrand-G\"odel
systems of equations)
and Stephen Cole Kleene (via the recursion and minimization
operators).}.
This approach was first announced by Kolmogorov in \cite{kolmo63},
1963, and then developped in \cite{kolmo65}, 1965.
Similar approaches were also independently developed
by Solomonoff in \cite{solo64a}, 1964,
and by Chaitin in \cite{chaitin66, chaitin69}, 1966-1969.
%
%
%------------------------------------------------------------------
\subsubsection{Digression on computability theory}
\label{sss:partial}
%------------------------------------------------------------------
%
The formalized notion of {\em computable function}
(also called recursive function)
goes along with that of {\em partial computable function}
(also called partial recursive function)
which should rather be called
{\em partially computable partial function},
i.e., the {\em partial} character has to be
distributed\footnote{In French, Daniel Lacombe
(\cite{lacombe60}, 1960) used the expression
{\em semi-fonction semi-r\'ecursive}.}.
\\
So, there are two theories :
\\\indent\textbullet~ {\em the theory of  computable functions},
\\\indent\textbullet~ {\em the theory of partial 
computable functions}.
\\
The ``right" theory, the one with a cornucopia of spectacular
results, is that of partial computable functions.
\begin{quote}
Let us pick up three fundamental results out of the cornucopia,
which we state in terms of computers and programming languages.
Let $\+I$ and $\+O$ be $\mathbb{N}$ or $A^*$
where $A$ is some finite or countably infinite alphabet
(or, more generally, $\+I$ and $\+O$ can be elementary sets,
cf. Definition~\ref{def:elementary}).
\begin{theorem}\label{thm:3thms}
\mbox{}\\
\mbox{}\\ {\em 1. [Enumeration theorem]}
The function which executes programs on their inputs:
\textnormal{(program, input)} $\to$ \textnormal{output}
is itself partial computable.
\\
Formally, this means that there exists a partial computable
function
$$
U: \{0,1\}^*\times\+I \to \+O
$$
such that the family of partial computable function
$\+I \to \+O$ is exactly
$\{U_e \mid e\in \{0,1\}^*\}$ where $U_e(x)=U(e,x)$.
\\
Such a function $U$ is called universal for partial
computable functions $\+I \to \+O$.
\medskip\\
{\em 2. [Parameter theorem (or $s^m_n$ thm)].}
One can exchange input and program
{\em(this is von Neumann's key idea for computers)}.
\\
Formally, this means that, letting $\+I=\+I_1\times \+I_2$,
universal maps $U_{\+I_1\times \+I_2}$ and $U_{\+I_2}$
are such that there exists a computable total map
$s: \{0,1\}^*\times \+I_1 \to \{0,1\}^*$
such that, for all $e\in \{0,1\}^*$,
$x_1\in \+I_1$ and $x_2\in \+I_2$,
$$
U_{\+I_1\times \+I_2}(e,(x_1,x_2)) = U_{\+I_2}(s(e,x_1),x_2)
$$
{\em 3. [Kleene fixed point theorem]}
For any transformation of programs, there is a program
which does the same {\em input $\to$ output} job
as its transformed program\footnote{This is the seed
of computer virology, cf. \cite{BKM06}}.
\\
Formally, this means that, for every partial computable map
$f: \{0,1\}^*\to \{0,1\}^*$, there exists $e$ such that
$$
\forall e\in \{0,1\}^*\quad \forall x\in\+I\quad
U(f(e),x)=U(e,x)
$$
\end{theorem}
\end{quote}
%
%
%------------------------------------------------------------------
\subsubsection{Kolmogorov complexity
(or program size complexity)}\label{sss:K}
%------------------------------------------------------------------
Turning to computability, the basic idea for Kolmogorov 
complexity\footnote{%
Delahaye's books 
\cite{{delahaye1999},{delahaye2006}} 
present a very attractive survey on Kolmogorov complexity.
                    }%
~can be summed up by the following equation:
\medskip\\\medskip\centerline{\em\begin{tabular}{|rcl|}
\hline
description &=& program
\\\hline
\end{tabular}}
When we say ``program", we mean a program taken from a family
of programs, i.e., written in a programming language or describing
a Turing machine or a system of Herbrand-G\"odel equations
or a Post system,...
\\
Since we are soon going to consider the length of programs,
following what has been said in \S\ref{sss:words},
we normalize programs: they will be binary words,
i.e., elements of $\{0,1\}^*$.
\\
So, we have to fix a function $\varphi:\{0,1\}^*\to\+O$ and
consider that the output of a program $p$ is $\varphi(p)$.
\\
Which $\varphi$ are we to consider? Since we know that there are
universal partial computable functions
(i.e., functions able to emulate any other partial computable
function modulo a computable transformation of programs,
in other words, a compiler from one language to another),
it is natural to consider universal partial computable functions.
Which agrees with what has been said in \S\ref{sss:partial}.
\\
Let us give the general definition of the Kolmogorov complexity
associated to any function $\{0,1\}^*\to\+O$.
\begin{definition}\label{def:Kphi}
If $\varphi:\{0,1\}^*\to\+O$ is a partial function,
set $K_\varphi : \+O \to \mathbb{N}$
$$
K_\varphi(y)=\min\{|p| : \varphi(p)=y\}
$$
with the convention that $\min\emptyset=+\infty$.
\\
Intuition:
$p$ is a program (with no input),
$\varphi$ executes programs
(i.e., $\varphi$ is altogether
a programming language
plus a compiler
plus a machinery to run programs)
and $\varphi(p)$ is the output of the run of program $p$.
Thus, for $y\in\+O$, $K_\varphi(y)$ is the length of shortest
programs $p$ with which $\varphi$ computes $y$
(i.e., $\varphi(p)=y$).
\end{definition}
\noindent
As said above, we shall consider this definition for
partial computable functions $\{0,1\}^*\to\+O$.
Of course, this forces to consider a set $\+O$ endowed with
a computability structure.
Hence the choice of sets that we shall call {\em elementary}
which do not exhaust all possible ones but will suffice for
the results mentioned in this paper.
\begin{definition}\label{def:elementary}
The family of elementary sets is obtained as follows:
\\- it contains $\mathbb{N}$ and the $A^*$'s where $A$ is a finite or
countable alphabet,
\\-  it is closed under finite (non empty) product,
product with any non empty finite set and
the finite sequence operator.
\end{definition}
\noindent{\bf Note.}
Closure under the finite sequence operator is used to encode
formulas in Theorem \ref{ss:godel}.

%
%------------------------------------------------------------------
\subsubsection{The invariance theorem}\label{sss:invariance}
%------------------------------------------------------------------
The problem with Definition \ref{def:Kphi} is that $K_\varphi$
strongly depends on $\varphi$.
Here comes a remarkable result, the invariance theorem,
which insures that {\em there is a smallest $K_\varphi$,
up to a constant}.
It turns out that the proof of this theorem only needs
the enumeration theorem and makes no use of the parameter
theorem (usually omnipresent in computability theory).
\begin{theorem}
[Invariance theorem, Kolmogorov, \cite{kolmo65}, 1965]
\label{thm:invariance}
Let $\+O$ be an elementary set (cf. Definition \ref{def:elementary}).
Among the $K_\varphi$'s, where $\varphi:\{0,1\}^*\to\+O$
varies in the family $PC^\+O$ of partial computable functions,
there is a smallest one, up to an additive constant
(= within some bounded interval). I.e.
$$
\exists V\in PC^\+O\quad \forall \varphi\in PC^\+O\quad
\exists c\quad \forall y\in\+O\quad
K_V(y) \leq K_\varphi(y) + c
$$
Such a $V$ is called optimal.
\\
Moreover, any universal partial computable function $\{0,1\}^*\to\+O$
is optimal.
\end{theorem}
\begin{quote}
{\em Proof.}
Let $U:\{0,1\}^*\times \{0,1\}^*\to\+O$ be partial computable and
universal for partial computable functions $\{0,1\}^*\to\+O$
(cf. point 1 of Theorem~\ref{thm:3thms}).
\\
Let $c: \{0,1\}^*\times\{0,1\}^*\to \{0,1\}^* $ be a total
computable injective map such that $|c(e,x)|=2|e|+|x|+1$
(cf. Proposition \ref{p:code}).
\\
Define $V: \{0,1\}^*\to\+O$, with domain included in the range of $c$,
as follows:
$$
\forall e\in\{0,1\}^*\ \forall x\in \{0,1\}^*\ \
V(c(e,x)) = U(e,x)
$$
where equality means that both sides are simultaneously
defined or not.
Then, for every partial computable function
$\varphi: \{0,1\}^*\to\+O$, for every $y\in\+O$,
if $\varphi=U_e$
(i.e., $\varphi(x)=U(e,x)$ for all $x$,
cf. point 1 of Theorem \ref{thm:3thms})
then
\begin{eqnarray*}
K_V(y) &=& \mbox{least $|p|$ such that $V(p)=y$}
\\
&\leq& \mbox{least $|c(e,x)|$ such that $V(c(e,x))=y$}
\\&&\hspace{1cm}
\mbox{(least is relative to $x$ since $e$~is~fixed)}
\\
&=& \mbox{least $|c(e,x)|$ such that $U(e,x))=y$}
\\
&=& \mbox{least $|x|+2|e|+1$ such that $\varphi(x)=y$}
\\&&\hspace{1cm}
\mbox{since $|c(e,x)|=|x|+2|e|+1$ and $\varphi(x)=U(e,x)$}
\\
& =& (\mbox{least $|x|$ such that $\varphi(x)=y$})+2|e|+1
\\
&=& K_\varphi(y)+2|e|+1 \hspace{6cm}\Box
\end{eqnarray*}
\end{quote}
Using the invariance theorem, the Kolmogorov complexity
$K^\+O:\+O\to\mathbb{N}$ is defined as $K_V$ where $V$ is any fixed optimal function.
The arbitrariness of the choice of $V$ does not modify
drastically $K_V$, merely up to a constant.
\begin{definition}
Kolmogorov complexity $K^\+O:\+O\to\mathbb{N}$ is
$K_V$, where $V$ is some fixed
optimal partial function $\{0,1\}^*\to\+O$.
When $\+O$ is clear from context, we shall simply write $K$.
\\
$K^\+O$ is therefore minimum among the $K_\varphi$'s,
up to an additive constant.
\\
$K^\+O$ is defined up to an additive constant:
if $V$ and $V'$ are both optimal then
$$
\exists c\quad \forall x\in\+O\quad |K_V(x) - K_{V'}(x)| \leq c
$$
\end{definition}
%
%
%------------------------------------------------------------------
\subsubsection{What Kolmogorov said about the constant}
\label{sss:constant}
%------------------------------------------------------------------
So Kolmogorov complexity is an integer defined up to
a constant\ldots !
But the constant is uniformly bounded for $x\in\+O$.
\\
Let us quote what Kolmogorov said about the constant
in \cite{kolmo65}, 1965:
\begin{quote}\em
Of course, one can avoid the indeterminacies
associated with the [above] constants,
by considering particular [\ldots functions $V$],
but it is doubtful that this can be done without
explicit arbitrariness.
\\
One must, however, suppose that the different
``reasonable" [above optimal functions] will lead to
``complexity estimates"
that will converge on hundreds of bits
instead of tens of thousands.
\\
Hence, such quantities as the ``complexity" of the text of
``War and Peace" can be assumed to be defined
with what amounts to uniqueness.
\end{quote}
In fact, this constant witnesses the multitude
of models of computation: universal Turing machines,
universal cellular automata,
Herbrand-G\"odel systems of equations, Post systems,
Kleene definitions,...
If we feel that one of them is canonical then we may consider
the associated Kolmogorov complexity as the right one
and forget about the constant.
This has been developed for Schoenfinkel-Curry combinators
$S,K,I$ by Tromp, cf. \cite{livitanyi} \S3.2.2--3.2.6.
\\
However, even if we fix a particular $K_V$, the importance
of the invariance theorem remains since it tells us that $K$
is less than {\em any} $K_\varphi$ (up to a constant).
A result which is applied again and again to develop the theory.
%
%
%------------------------------------------------------------------
\subsubsection{Considering inputs:
conditional Kolmogorov complexity}
%------------------------------------------------------------------
In the enumeration theorem,
we considered {\em(program, input) $\to$ output} functions
(cf. Theorem \ref{thm:3thms}).
Then, in the definition of Kolmogorov complexity, we gave up
the inputs, dealing with {\em program $\to$ output} functions.
\\
Conditional Kolmogorov complexity deals with the inputs. 
Instead of measuring the information content of $y\in\+O$,
we measure it given as free some object $z$, which may help
to compute $y$.
A trivial case is when $z=y$, then the information content
of $y$ given $y$ is null. In fact, there is an obvious program
which outputs exactly its input, whatever be the input.
\\
Let us mention that, in computer science, inputs are also
considered as {\em  environments}.
\\
Let us state the formal definition and the adequate
invariance theorem.
\begin{definition}\label{def:condKphi}
If $\varphi:\{0,1\}^*\times\+I\to\+O$ is a partial function,
set $K_\varphi(\ \mid\ ) : \+O\times\+I \to \mathbb{N}$
$$
K_\varphi(y \mid z)=\min\{|p| \mid \varphi(p,z)=y\}
$$
Intuition:
$p$ is a program (with expects an input $z$),
$\varphi$ executes programs
(i.e., $\varphi$ is altogether
a programming language
plus a compiler
plus a machinery to run programs)
and $\varphi(p,z)$ is the output of the run of program $p$
on input $z$.
Thus, for $y\in\+O$, $K_\varphi(y \mid z)$ is the length of shortest
programs $p$ with which $\varphi$ computes $y$ on input $z$
(i.e., $\varphi(p,z)=y$).
\end{definition}
\begin{theorem}[Invariance theorem for conditional complexity]
Among the $K_\varphi(\ |\ )$'s, where $\varphi $ varies
in the family $PC^\+O_\+I$
of partial computable functions $\{0,1\}^*\times \+I\to\+O$,
there is a smallest one, up to an additive constant
(i.e., within some bounded interval) :
$$
\exists V\in PC^\+O_\+I\quad \forall \varphi\in PC^\+O_\+I\quad
\exists c\quad \forall y\in\+O\quad \forall z\in\+I\quad
K_V(y\mid z) \leq K_\varphi(y\mid z) + c
$$
Such a $V$ is called optimal.
\\
Moreover, any universal partial computable map
$\{0,1\}^*\times \+I\to\+O$ is optimal.
\end{theorem}
\noindent
The proof is similar to that of Theorem \ref{thm:invariance}.
\begin{definition}
$K^{\+I\to\+O}:\+O\times\+I\to\mathbb{N}$ is $K_V(\ |\ )$ where $V$
is some fixed optimal partial function.
\medskip\\
$K^{\+I\to\+O}$ is defined up to an additive constant:
if $V$ et $V'$ are both optimal then
$$
\exists c\quad \forall y\in\+O\quad \forall z\in\+I\quad
|K_V(y\mid z) - K_{V'}(y\mid z)| \leq c
$$
\end{definition}
\noindent
Again, an integer defined up to a constant\ldots!
However, the constant is uniform in $y\in\+O$ and $z\in\+I$.
%
%
%------------------------------------------------------------------
\subsubsection{Simple upper bounds for Kolmogorov complexity}
%------------------------------------------------------------------
%
Finally, let us mention rather trivial upper bounds:
\\- the information content of a word is at most its length.
\\- conditional complexity cannot be harder than the non conditional
one.
\begin{proposition}\label{p:bound}
\mbox{}\\
\mbox{}\\ 1. There exists $c$ such that
$$
\forall x\in\{0,1\}^*\ \ K^{\{0,1\}^*}(x)\leq |x|+c
\quad,\quad
\forall n\in\mathbb{N}\ \ K^\mathbb{N}(n)\leq \log(n)+c
$$
2. There exists $c$ such that
$$
\forall x\in\+O\quad \forall y\in\+I
\quad K^{\+I\to\+O}(x\mid y)\leq K^\+O(x)+c
$$
3. Let $f:\+O\to\+O'$ be computable. There exists $c$ such that
\begin{center}
\begin{tabular}{lrcl}
$\forall x\in\+O$ &$K^{\+O'}(f(x))$&$\leq$&$K^\+O(x) +c$
\\
$\forall x\in\+O\quad \forall Y\in\+I$&
$K^{\+I\to\+O'}(f(x)\mid y)$&$\leq$&$K^{\+I\to\+O}(x\mid y) +c$
\end{tabular}
\end{center}
\end{proposition}
\begin{quote}
\begin{proof}
We only prove 1.
Let $Id:\{0,1\}^*\to\{0,1\}^*$ be the identity function.
The invariance theorem insures that there exists $c$ such that
$K^{\{0,1\}^*} \leq K^{\{0,1\}^*}_{Id} + c$.
Now, it is easy to see that $K^{\{0,1\}^*}_{Id}=|x|$, so that
$K^{\{0,1\}^*}(x)\leq |x|+c$.
\\
Let $\theta:\{0,1\}^*\to\mathbb{N}$ be the function 
(which is, in fact, a bijection)
which associates to a word $u=a_{k-1}...a_0$ the integer
$$
\theta(u) = (2^k+a_{k-1}2^{k-1}+...+2a_1+a_0) -1
$$
(i.e., the predecessor of the integer with binary representation $1u$).
Clearly,
$K^\mathbb{N}_\theta(n)=\lfloor \log(n+1)\rfloor$.
The invariance theorem insures that there exists $c$ such that
$K^\mathbb{N} \leq K^\mathbb{N}_\theta + c$. 
Hence $K^\mathbb{N}(n)\leq \log(n)+c+1$
for all $n\in\mathbb{N}$.
\end{proof}
\end{quote}
The following technical property is a variation of an argument
already used in \S\ref{sss:closer}: the rank of an element
in a set defines this element, and if the set is computable,
so is this process.
\begin{proposition}\label{p:rank}
Let $A\subseteq \mathbb{N}\times\+O$ be computable such that
$A_n = A\cap(\{n\}\times\+O)$ is finite for all $n$.
Then, letting $\sharp X$ be the number of elements of $X$,
$$
\exists c\quad \forall x\in A_n\quad
K(x\mid n) \leq \log(\sharp(A_n))+c
$$
\end{proposition}
\begin{quote}
\begin{proof}
Observe that $x$ is determined by its rank in $A_n$.
This rank is an integer $<\sharp A_n$ hence its binary
representation has length $\leq\lfloor\log(\sharp A_n)\rfloor+1$.
\end{proof}
\end{quote}
%
%
%
%
%
%
%
%%%%%%%%%%%%%%%%%%%%%%%%%%%%%%%%%%%
%%%%%%%%%%%%%%%%%%%%%%%%%%%%%%%%%%%
%%%%%%%%%%%%%%%%%%%%%%%%%%%%%%%%%%%
\section[Kolmogorov complexity and undecidability]
         {Kolmogorov complexity and undecidability}
%%%%%%%%%%%%%%%%%%%%%%%%%%%%%%%%%%%
%%%%%%%%%%%%%%%%%%%%%%%%%%%%%%%%%%%
%%%%%%%%%%%%%%%%%%%%%%%%%%%%%%%%%%%
%
%
%%%%%%%%%%%%%%%%%%%%%%%%%%%%%%%%%%%
\subsection{$K$ is unbounded}
%%%%%%%%%%%%%%%%%%%%%%%%%%%%%%%%%%%
Let $K=K_V : \+O\to\mathbb{N}$ where $V:\{0,1\}^*\to\+O$ is optimal
(cf. Theorem \S\ref{thm:invariance}).
Since there are finitely many programs of size $\leq n$
(namely, the $2^{n+1}-1$ binary words of size $\leq n$),
there are finitely many elements of $\+O$ with Kolmogorov complexity
less than $n$.
This shows that $K$ is unbounded.
%
%%%%%%%%%%%%%%%%%%%%%%%%%%%%%%%%%%%
\subsection{$K$ is not computable}\label{ss:Knoncomput}
%%%%%%%%%%%%%%%%%%%%%%%%%%%%%%%%%%%
Berry's paradox (cf. \S\ref{sss:berry}) has a counterpart in terms
of Kolmogorov complexity: it gives a very simple proof that $K$,
which is a total function $\+O\to\mathbb{N}$, is not computable.
\begin{quote}
{\em Proof that $K$ is not computable.}
For simplicity of notations, we consider the case $\+O=\mathbb{N}$.
Define $L:\mathbb{N}\to\+O$ as follows:
\begin{eqnarray*}
L(n) &=& \mbox{least  $k$ such that $K(k)\geq 2n$}
\end{eqnarray*}
So that $K(L(n))\geq 2n$ for all $n$.
If $K$ were computable so would be $L$.
Let $V:\+O\to\mathbb{N}$ be optimal, i.e., $K=K_V$.
The invariance theorem insures that there exists $c$ such that
$K\leq K_L+c$. Observe that $K_L(L(n)\leq n$ by definition of
$K_L$.
Thus,
$$
2n \leq K(L(n)) \leq K_L(L(n)+c \leq n+c
$$
A contradiction for $n>c$.
\hfill{$\Box$}
\end{quote}
The non computability of $K$ can be seen as a version of the
undecidability of the halting problem.
In fact, there is a simple way to compute $K$ when the halting
problem is used as an oracle.
To get the value of $K(x)$, proceed as follows:
\\\indent- enumerate the programs in $\{0,1\}^*$ in lexicographic
          order,
\\\indent- for each program $p$, check if $V(p)$ halts
          (using the oracle),
\\\indent- in case $V(p)$ halts then compute its value,
\\\indent- halt and output $|p|$ when some $p$ is obtained
           such that $V(p)=x$.
\medskip\\
The converse is also true: one can prove that
{\em the halting problem is computable with $K$ as an oracle}.
\medskip\\
The argument for the undecidability of $K$ can be used to prove
a much stronger statement: $K$ can not be bounded from below by
any unbounded partial computable function.
\begin{theorem}[Kolmogorov]
There is no unbounded partial recursive function
$\psi:\+O\to\mathbb{N}$ such that $\psi(x)\leq K(x)$ for all $x$
in the domain of $\psi$.
\end{theorem}
\noindent
Of course, $K$ is bounded from above by a total computable
function, cf. Proposition \ref{p:bound}.
%
%%%%%%%%%%%%%%%%%%%%%%%%%%%%%%%%%%%
\subsection{$K$ is computable from above}\label{ss:Kabove}
%%%%%%%%%%%%%%%%%%%%%%%%%%%%%%%%%%%
%
Though $K$ is not computable, it can be approximated from above.
The idea is simple. Suppose $\+O=\{0,1\}^*$.
Let $c$ be as in point 1 of Proposition \ref{p:bound}.
Consider all programs of length less than $|x|+c$
and let them be executed during $t$ steps.
If none of them converges and outputs $x$
then take $|x|+c$ as a $t$-bound.
If some of them converges and outputs $x$
then the bound is the length of the shortest such program.
\\
The limit of this process is $K(x)$, it is obtained at some
finite step which we are not able to bound.
\\
Formally, this means that there is some 
$F:\+O\times\mathbb{N}\to\mathbb{N}$
which is computable and decreasing in its second argument
such that
$$
K(x) = \lim_{t\to+\infty} F(x,t) = \min\{F(x,t) \mid t\in\mathbb{N}\}
$$
%
%
%%%%%%%%%%%%%%%%%%%%%%%%%%%%%%%%%%%
\subsection{Kolmogorov complexity
            and G\"odel's incompleteness theorem}\label{ss:godel}
%%%%%%%%%%%%%%%%%%%%%%%%%%%%%%%%%%%
A striking version of G\"odel's incompleteness theorem has been
given by Chaitin in \cite{chaitin71, chaitin74}, 1971-1974,
in terms of Kolmogorov complexity.
Since G\"odel's celebrated proof of the incompleteness theorem,
we know that, in the language of arithmetic,
one can formalize computability and logic.
In particular, one can formalize Kolmogorov complexity and
statements about it.
Chaitin's proves a version of the incompleteness theorem
which insures that among true unprovable formulas
there are all true statements $K(u)>n$ for $n$ large enough.
\begin{theorem}[Chaitin, \cite{chaitin74}, 1974]
\label{thm:godel1}
Let $\+T$ be a computably enumerable set of axioms
in the language of arithmetic.
Suppose that all axioms in $\+T$ are true in the standard model
of arithmetics with base $\mathbb{N}$.
Then there exists $N$ such that if $\+T$ proves
$K(u)>n$ (with $u\in\{0,1\}^*$ and $n\in\mathbb{N}$) then $n\leq N$.
\end{theorem}
\noindent
How the constant $N$ depends on $\+T$ has been giving a
remarkable analysis by Chaitin.
To that purpose, he extends Kolmogorov complexity to
computably enumerable sets.
\begin{definition}[Chaitin, \cite{chaitin74}, 1974]
Let $\+O$ be an elementary set
(cf. Definition \ref{def:elementary})
and $\+C\+E$ be the family of computably enumerable (c.e.)
subsets of $\+O$.
To any partial computable $\varphi:\{0,1\}^*\times\mathbb{N}\to\+O$,
associate the Kolmogorov complexity $K_\varphi:\+C\+E\to\mathbb{N}$
such that, for all c.e. subset $\+T$ of $\+O$,
$$
K_\varphi(\+T)
= \min\{|p| \mid \+T=\{\varphi(p,t)\mid t\in\mathbb{N}\}\}
$$
(observe that $\{\varphi(p,t)\mid t\in\mathbb{N}\}$ is always c.e.
and any c.e. subset of $\+O$ can be obtained in this way
for some $\varphi$).
\end{definition}
\noindent
The invariance theorem still holds for this notion of
Kolmogorov complexity, leading to the following notion.
\begin{definition}[Chaitin, \cite{chaitin74}, 1974]
$K^{\+C\+E}:\+C\+E\to\mathbb{N}$ is $K_\varphi$ where $\varphi$
is some fixed optimal partial function.
It is defined up to an additive constant.
\end{definition}
\noindent
We can now state how the constant $N$ in Theorem \ref{thm:godel1}
depends on the theory $\+T$.
\begin{theorem}[Chaitin, \cite{chaitin74}, 1974]
\label{thm:godel2}
There exists a constant $c$ such that, for all c.e. sets $\+T$
satisfying the hypothesis of Theorem \ref{thm:godel1},
the associated constant $N$ is such that
$$ N\leq K^{\+C\+E}(\+T)+c $$
\end{theorem}
\noindent
Chaitin also reformulates  Theorem \ref{thm:godel1}
as follows:
\begin{quote}{\em
If $\+T$ consist of true formulas then it cannot prove that
a string has Kolmogorov complexity greater than the
Kolmogorov complexity of $\+T$ itself
(up to a constant independent of $\+T$).}
\end{quote}
\begin{remark}
The previous statement, and Chaitin's assertion that the
Kolmogorov complexity of $\+T$ somehow measures the power
of $\+T$ as a theory, has been much criticized in
van Lambalgen (\cite{lambalgen}, 1989),
Fallis (\cite{fallis},  1996)
and Raatikainen (\cite{Raatikainen}, 1998).
Raatikainen's main argument in \cite{Raatikainen} against Chaitin's
interpretation is that the constant in Theorem \ref{thm:godel1}
strongly depends on the choice of the optimal function $V$
such that $K=K_V$.
Indeed, for any fixed theory $\+T$, one can choose such a $V$
so that the constant is zero!
And also choose $V$ so that the constant is arbitrarily large.
\\
Though these arguments are perfectly sound,
we disagree with the criticisms issued from them.
Let us detail three main rebuttals.
\medskip
\\\textbullet~
First, such arguments are based on the use of optimal functions
associated to very unnatural universal functions $V$
(cf. point 1 of Theorem \ref{thm:3thms} and the last assertion of
Theorem~\ref{thm:invariance}).
It has since been recognized that universality is not always
sufficient to get smooth results. Universality by prefix adjunction
is sometimes required,
(cf., for instance, \S2.1 and \S6 in Becher, Figueira, Grigorieff \& Miller, 2006).
This means that, for an enumeration $(\varphi_e)_{e\in\words}$
of partial computable functions, the optimal function $V$
is to satisfy equality $V(ep)=\varphi_e(p)$, for all $e,p$,
where $ep$ is the concatenation of the strings $e$ and $p$.
\medskip
\\\textbullet~
Second, and more important than the above technical
counterargument,
it is a simple fact that modelization rarely rules out all
pathological cases.
It is intended to be used in ``reasonable" cases.
Of course, this may be misleading,
but perfect modelization is illusory.
In our opinion, this is best illustrated by Kolmogorov's citation
quoted in \S\ref{sss:constant}
to which Raatikainen's argument could be applied mutatis mutandis:
there are optimal functions for which the complexity of the text of
``War and Peace" is null and other ones for which it is arbitrarily large.
Nevertheless, this does not prevent Kolmogorov to assert
(in the founding paper of the theory \cite{kolmo65}):
{\em [For] ``reasonable" [above optimal functions],
such quantities as the ``complexity" of the text of
``War and Peace" can be assumed to be defined
with what amounts to uniqueness.}
\medskip
\\\textbullet~
Third, a final technical answer to such criticisms has been recently
provided by Calude \& Jurgensen in \cite{CaludeJurgensen},  2005.
They improve the incompleteness result given by
Theorem~\ref{thm:godel1},
proving that, for a class of formulas in the vein of those in that
theorem,
the probability that such a formula of length $n$ is provable
tends to zero when $n$ tends to infinity
whereas the probability that it be true has a strictly positive
lower bound. 
\end{remark}
%
%
%
%%%%%%%%%%%%%%%%%%%%%%%%%%%%%%%%%%%
%%%%%%%%%%%%%%%%%%%%%%%%%%%%%%%%%%%
%%%%%%%%%%%%%%%%%%%%%%%%%%%%%%%%%%%
\section{Kolmogorov complexity: some variations}
%%%%%%%%%%%%%%%%%%%%%%%%%%%%%%%%%%%
%%%%%%%%%%%%%%%%%%%%%%%%%%%%%%%%%%%
%%%%%%%%%%%%%%%%%%%%%%%%%%%%%%%%%%%
%
\noindent
\\{\bf Note.}
The denotations of (plain) Kolmogorov complexity
(that of \S \ref{sss:invariance})
and its prefix version (cf. \ref{ss:self})
may cause some confusion.
They long used to be respectively denoted by $K$ and $H$
in the literature.
But in their book \cite{livitanyi}  (first edition, 1993),
Li \& Vitanyi respectively
denoted them by $C$ and $K$.
Due to the large success of this book, these last denotations
are since used in many papers.
So that two incompatible denotations now appear in the
literature.
In this paper,
we stick to the traditional denotations $K$ and $H$.
%
%
%%%%%%%%%%%%%%%%%%%%%%%%%%%%%%%%%%%
\subsection{Levin monotone complexity}
\label{ss:monotone}
%%%%%%%%%%%%%%%%%%%%%%%%%%%%%%%%%%%
%
Kolmogorov complexity is non monotone, be it
on $\mathbb{N}$ with the natural ordering
or on $\{0,1\}^*$ with the lexicographic ordering.
In fact, for every $n$ and $c$,
there are strings of length $n$ with complexity $\geq n(1-2^{-c})$
(cf. Proposition \ref{p:incompressible}).
However, since $n\mapsto 1^n$ is computable,
$K(1^n)\leq K(n)+O(1)\leq\log n+O(1)$
(cf. point 3 of Proposition \ref{p:bound})
is much less than $n(1-2^{-c})$ for $n$ large enough.
\medskip
\\Leonid Levin (\cite{levin73}, 1973) introduced a monotone version
of Kolmogorov complexity.
The idea is to consider possibly infinite computations of
Turing machines which never erase anything on the output tape.
Such machines have finite or infinite outputs and
compute total maps $\{0,1\}^*\to\mix$ where
$\mix=\{0,1\}^*\cup\cantor$ is the family of finite or infinite
binary strings. These maps can also be viewed as limit maps
$p\to\sup_{t\to\infty}\varphi(p,t)$ where
$\varphi:\{0,1\}^*\times\mathbb{N}\to\{0,1\}^*$
is total monotone non decreasing
in its second argument.
\\
To each such map $\varphi$, Levin associates a monotone non
decreasing map $K_\varphi^{mon}:\{0,1\}^*\to\mathbb{N}$ such that
$$
K_\varphi^{mon}(x)
=\min\{|p|\mid \exists t\ x\leq_{pref}\varphi(p,t)\}
$$
\begin{theorem}[Levin (\cite{levin73}, 1973]
\mbox{}\\
\mbox{}\\ 1. If $\varphi$ is total computable and monotone non decreasing
in its second argument then $K_\varphi^{mon}:\{0,1\}^*\to\mathbb{N}$
is monotone non decreasing:
$$
 x\leq_{pref} y \Rightarrow
 K_\varphi^{mon}(x)\leq K_\varphi^{mon}(y)
$$
2. Among the $K_\varphi^{mon}$'s, $\varphi$ total computable
monotone non decreasing in its second argument,
there exists a smallest one, up to a constant.
\end{theorem} 
\noindent
Considering total $\varphi$'s in the above theorem is a priori
surprising since there is no computable enumeration of total
computable functions and the proof of the Invariance
Theorem \ref{thm:invariance} is based on the enumeration theorem
(cf. Theorem~\ref{thm:3thms}).
The trick to overcome that problem is as follows.
\begin{itemize}
\item [\textbullet]
Consider all partial computable
$\varphi:\{0,1\}^*\times\mathbb{N}\to\{0,1\}^*$
which are total monotone non decreasing in their second argument.
\item [\textbullet]
Associate to each such $\varphi$ a total $\widetilde{\varphi}$
defined as follows:
$\widetilde{\varphi}(p,t)$ is the largest $\varphi(p,t')$
such that $t'\leq t$ and $\varphi(t')$ is defined within $t+1$
computation steps if there is such a $t'$.
If there is none then $\widetilde{\varphi}(p,t)$ is the empty word.
\item [\textbullet]
Observe that
$K_\varphi^{mon}(x)=K_{\widetilde{\varphi}}^{mon}(x)$.
\end{itemize}
In \S\ref{sss:SKmH}, we shall see some remarkable property of Levin
monotone complexity $K^{mon}$
concerning Martin-L\"of random reals.
%
%
%%%%%%%%%%%%%%%%%%%%%%%%%%%%%%%%%%%
\subsection{Schnorr process complexity}
\label{ss:Schnorr}
%%%%%%%%%%%%%%%%%%%%%%%%%%%%%%%%%%%
%
Another variant of Kolmogorov complexity has been introduced by
Klaus Peter Schnorr in \cite{schnorr73}, 1973.
It is based on the subclass of partial computable functions
$\varphi:\{0,1\}^*\to\{0,1\}^*$ which are
monotone non decreasing relative to the prefix ordering:
\begin{center}
(*)\qquad
$(p\leq_{pref} q\ \wedge\ \varphi(p),\varphi(q)$ are both defined)
$~\Rightarrow~\varphi(p)\leq_{pref}\varphi(q)$
\end{center}
Why such a requirement on $\varphi$?
The reason can be explained as follows.
\begin{itemize}
\item [\textbullet]
Consider a sequential composition (i.e., a pipeline)
of two processes, formalized as two functions $f,g$.
The first one takes an input $p$ and outputs $f(p)$,
the second one takes $f(p)$ as input and outputs $g(f(p))$.
\item [\textbullet]
Each process is supposed to be monotone:
the first letter of $f(p)$ appears first, then the second one, etc.
Idem with the digits of $g(q)$ for any input $q$. 
\item [\textbullet]
More efficiency is obtained if one can develop the computation
of $g$ on input $f(p)$ as soon as the letters of $f(p)$ appear.
More precisely, suppose the prefix $q$ of $f(p)$ has already
appeared but there is some delay to get the subsequent letters.
Then we can compute $g(q)$.
But this is useful only in case the computation of $g(q)$
is itself a prefix of that of $g(f(p))$.
This last condition is exactly the requirement $(*)$.
\end{itemize}
\noindent
An enumeration theorem holds for the $\varphi$'s satisfying $(*)$,
allowing to prove an invariance theorem and to define a so-called
process complexity $K^{proc}:\{0,1\}^*\to\mathbb{N}$.
The same remarkable property of Levin's monotone complexity also
holds with Schnorr process complexity, cf. \S\ref{sss:SKmH}.
%
%
%%%%%%%%%%%%%%%%%%%%%%%%%%%%%%%%%%%
\subsection{Prefix (or self-delimited) complexity}
\label{ss:self}
%%%%%%%%%%%%%%%%%%%%%%%%%%%%%%%%%%%
%
Levin (\cite{levin74}, 1974), G\'acs (\cite{gacs74}, 1974)
and Chaitin (\cite{chaitin75}, 1975)
introduced the most successful variant of Kolmogorov complexity:
the prefix complexity.
The idea is to restrict the family of partial computable functions
$\{0,1\}^*\to\+O$
(recall $\+O$ denotes an elementary set in the sense of
Definition \ref{def:elementary})
to those which have prefix-free domains, i.e.
any two words in the domain are incomparable with respect to the
prefix ordering.
\medskip
\\
An enumeration theorem holds for the $\varphi$'s satisfying $(*)$,
allowing to prove an invariance theorem and to define the so-called
prefix complexity $H:\{0,1\}^*\to\mathbb{N}$ 
(not to be confused with the
entropy of a family of frequencies, cf. \S\ref{sss:entropy}).
\begin{theorem}
Among the $K_\varphi$'s, where $\varphi:\{0,1\}^*\to\+O$ varies over
partial computable functions with prefix-free domain,
there exists a smallest one, up to a constant.
This smallest one (defined up to a constant),
denoted by $H^\+O$, is called the prefix complexity.
\end{theorem} 
\noindent
This prefix-free condition on the domain may seem rather technical.
A conceptual meaning of this condition has been given by Chaitin
in terms of self-delimitation.
\begin{proposition}[Chaitin, \cite{chaitin75}, 1975]
A partial computable function $\varphi:\{0,1\}^*\to\+O$
has prefix-free domain if and only if it can be computed
by a Turing machine $\+M$ with the following property:
\begin{quote}
If $x$ is in domain$(\varphi)$
(i.e., $\+M$ on input $p$ halts in an accepting state at some step)
then the head of the input tape of $\+M$ reads entirely
the input $p$ but never moves to the cell right to $p$.
\end{quote}
\end{proposition}
\noindent
This means that $p$, interpreted as a program, has no need of
external action (as that of an end-of-file symbol)
to know where it ends: as Chaitin says, the program is
self-delimited.
A comparison can be made with biological phenomena.
For instance, the hand of a person grows during its childhood
and then stops growing.
No external action prevents the hand to go on growing.
There is something inside the genetic program which creates
a halting signal so that the hand stops growing.
\medskip
\\
The main reason for the success of the prefix complexity is that,
with prefix-free domains, one can use the Kraft-Chaitin inequality
(cf. the proof of Theorem \ref{thm:H} in \S\ref{sss:sourcecoding})
and get remarkable properties.
\begin{theorem}[Kraft-Chaitin inequality]
A sequence (resp. computable sequence) $(n_i)_{i\in\mathbb{N}}$
of non negative integers
is the sequence of lengths of a prefix-free (resp. computable)
family of words $(u_i)_{i\in\mathbb{N}}$
if and only if $\sum_{i\in\mathbb{N}}2^{-n_i}\leq1$.
\end{theorem}
Let us state the most spectacular property of the prefix complexity.
\begin{theorem}[The Coding Theorem (Levin (\cite{levin74}, 1974)]
\label{thm:coding}
Consider the family $\ell_1^{c.e.}$ of sequences
of non negative real numbers $(r_x)_{x\in\+O}$ such that
\begin{itemize}
\item [\textbullet]
$\sum_{x\in\+O}r_x <+\infty$ (i.e., the series is summable),
\item [\textbullet]
$\{(x,q)\in\+O\times\mathbb{Q} \mid q<r_x\}$ is computably enumerable
(i.e., the $r_x$'s have c.e. left cuts in the set of rational
numbers $\mathbb{Q}$ and this is uniform in $x$).
\end{itemize}
The sequence $(2^{-H^\+O(x)})_{x\in\+O}$ is in
$\ell_1^{c.e.}$ and, up to a multiplicative factor,
it is the largest sequence in $\ell_1^{c.e.}$.
This means that
$$
\forall (r_x)_{x\in\+O}\in\ell_1^{c.e.}\quad
\exists c\quad
\forall x\in\+O\quad r_x\leq c\ 2^{-H^\+O(x)}
$$
\end{theorem}
\noindent
In particular, consider a countably infinite alphabet $A$.
Let $V:\{0,1\}^*\to A$ be a partial computable function with
prefix-free domain such that $H^A=K_V$.
Consider the prefix code $(p_a)_{a\in A}$ such that,
for each letter $a\in A$, $p_a$ is a shortest binary string
such that $V(p_a)=a$.
Then, for every probability distribution $P:A\to[0,1]$ over
the letters of the alphabet $A$,
which is computably approximable from below
(i.e., $\{(a,q)\in A\times\mathbb{Q} \mid q<P(a)\}$ is computably
enumerable), we have
$$
\forall a\in A\quad P(a)\leq c\ 2^{-H^A(a)}
$$
for some $c$
which depends on $P$ but not on $a\in A$.
This inequality is the reason why the sequence
$(2^{-H^A(a)})_{a\in A}$ is also called
{\em the universal a priori probability}
(though, strictly speaking, it is not a probablity
since the $2^{-H^A(a)}$'s do not sum up to $1$).
%
%
%
%%%%%%%%%%%%%%%%%%%%%%%%%%%%%%%%%%%
\subsection{Oracular Kolmogorov complexity }
%%%%%%%%%%%%%%%%%%%%%%%%%%%%%%%%%%%
%
As is always the case in computability theory, everything
relativizes to any oracle $Z$.
Relativization modifies the equation given at the start of \S\ref{sss:K}, which is now
\medskip\\ \medskip\centerline{\em\begin{tabular}{rcl}
description &=& program of a partial $Z$-computable function
\end{tabular}}
and for each possible oracle $Z$ there exists a Kolmogorov
complexity relative to oracle $Z$.
\medskip\\
Oracles in computability theory can also be considered as
second-order arguments of computable or partial computable
{\em functionals}.
The same holds with oracular Kolmogorov complexity:
the oracle $Z$ can be seen as a second-order condition
for a {\em second-order conditional Kolmogorov complexity}
$$
K(y\mid Z)\hspace{3mm}\mbox{where}\hspace{3mm}
K(\ \mid\ ):\+O \times P(\+I) \to \mathbb{N}
$$
Which has the advantage that the unavoidable constant in the
``up to a constant" properties does not depend on the
particular oracle. It depends solely on the considered
functional.
\\
Finally, one can mix first-order and second-order conditions,
leading to a conditional Kolmogorov complexity with both
first-order and second-order conditions
$$
K(y\mid z, Z)\hspace{3mm}\mbox{where}\hspace{3mm}
K(\ \mid\ ,\ ):\+O \times \+I\times P(\+I) \to \mathbb{N}
$$
We shall see in \S\ref{sss:nies}
an interesting property involving oracular Kolmogorov complexity. 
%
%
%
%%%%%%%%%%%%%%%%%%%%%%%%%%%%%%%%%%%
\subsection{Sub-oracular Kolmogorov complexity}
\label{ss:Kinfinite}
%%%%%%%%%%%%%%%%%%%%%%%%%%%%%%%%%%%
%
Going back to the idea of possibly infinite computations
as in \S\ref{ss:monotone}, Let us define
$K^\infty:\{0,1\}^*\to\mathbb{N}$ such that
$$
K^\infty(x) = \min\{|p| \mid U(p)=x\}
$$
where $U$ is the map $\{0,1\}^*\to\mix$ computed by a universal
Turing machine with possibly infinite computations.
This complexity lies between $K$ and $K(\ \mid\emptyset')$
(where $\emptyset'$ is a computably enumerable set which encodes
the halting problem):
$$
\forall x\quad
K(x\mid\emptyset')\leq K^\infty(x) +O(1)\leq K(x) +O(1)
$$
This complexity is studied in
\cite{becherfigniespicchi05}, 2005, by Becher, Figueira, Nies \& Picci,
and also in our paper \cite{ferbusgrigoKandAbstraction1}, 2006.
%
%
%
% 
%
%
%
%%%%%%%%%%%%%%%%%%%%%%%%%%%%%%%%%%%
%%%%%%%%%%%%%%%%%%%%%%%%%%%%%%%%%%%
%%%%%%%%%%%%%%%%%%%%%%%%%%%%%%%%%%%
\section{Formalization of randomness: finite objects}
%%%%%%%%%%%%%%%%%%%%%%%%%%%%%%%%%%%
%%%%%%%%%%%%%%%%%%%%%%%%%%%%%%%%%%%
%%%%%%%%%%%%%%%%%%%%%%%%%%%%%%%%%%%
%
%
%%%%%%%%%%%%%%%%%%%%%%%%%%%%%%%%%%%
\subsection{Sciences of randomness: probability theory}
%%%%%%%%%%%%%%%%%%%%%%%%%%%%%%%%%%%
%
Random objects {\em(words, integers, reals,...)}
constitute the basic intuition for probabilities
{\em... but they are not considered per se.}
No formal definition of random object is given:
there seems to be no need for such a formal concept.
The existing formal notion of {\em random variable} has nothing
to do with randomness: a random variable is merely a
{\em measurable function} which can be as non random as one likes.
\medskip\\
It sounds strange that the mathematical theory which deals with
randomness removes the natural basic questions:
\\\indent\textbullet~ {\em What is a random string?}
\\\indent\textbullet~  {\em What is a random infinite sequence?}
\\
When questioned, people in probability theory agree that they
skip these questions but do not feel sorry about it.
As it is, the theory deals with laws of randomness and
is so successful that it can do without entering this problem.
\medskip\\
This may seem to be analogous to what is the case in geometry.
What are points, lines, planes?
No definition is given, only relations between them.
Giving up the quest for an analysis of the nature of geometrical
objects  in profit of the axiomatic method
has been a considerable scientific step.
\\
However, we contest such an analogy.
Random objects are heavily used in many areas of science
and technology: sampling, cryptology,...
Of course, such objects are in fact
{\em ``as much as we can random"}.
Which means {\em fake randomness}.
But they refer to an ideal notion of randomness which cannot be
simply disregarded.
\medskip\\
In fact, since Pierre Simon de Laplace (1749--1827),
some probabilists never gave up the idea of formalizing
the notion of random object.
Let us cite particularly Richard von Mises (1883--1953)
and Kolmogorov.
In fact, it is quite impressive that, having so brilliantly
and efficiently axiomatized probability theory via measure theory
in \cite{kolmo33}, 1933, Kolmogorov was not fully satisfied
of such foundations\footnote{
Kolmogorov is one of the rare probabilists -- up to now --
not to believe that Kolmogorov's axioms for probability theory
do not constitute the last word about formalizing randomness...}.
And he kept a keen interest to the quest for a formal
notion of randomness initiated by von Mises in the 20's.
%
%
%%%%%%%%%%%%%%%%%%%%%%%%%%%%%%%%%%%
\subsection{The 100 heads paradoxical result in probability theory}
%%%%%%%%%%%%%%%%%%%%%%%%%%%%%%%%%%%
%
That probability theory fails to completely account for randomness
is strongly witnessed by the following paradoxical fact.
In probability theory,
{\em if we toss an unbiaised coin 100 times then
       100 heads are just as probable as any other outcome!}
Who really believes that?
{\em\begin{quote}
The axioms of probability theory, as developped by Kolmogorov,
do not solve all mysteries that they are sometimes supposed to.

\hfill{G\' acs, \cite{gacs93}, 1993}
\end{quote}}

%
%%%%%%%%%%%%%%%%%%%%%%%%%%%%%%%%%%%
\subsection{Sciences of randomness: cryptology}
%%%%%%%%%%%%%%%%%%%%%%%%%%%%%%%%%%%
%
Contrarily to probability theory, cryptology heavily uses random
objects.
Though again, no formal definition is given, random sequences
are produced which are not fully random, just hard enough so that
the mechanism which produces them cannot be discovered in reasonable
time.
\begin{quote}{\em
Anyone who considers arithmetical methods of producing random reals
is, of course, in a state of sin.
For, as has been pointed out several times, there is no such thing
as a random number --- there are only methods to produce random
numbers, and a strict arithmetical procedure is of course not such
a method.}

\hfill{Von Neumann, \cite{neumannsin}, 1951}
\end{quote}
So, what is ``true" randomness?
Is there something like a degree of randomness?
Presently, (fake) randomness only means to pass
some statistical tests.
One can ask for more.
%
%%%%%%%%%%%%%%%%%%%%%%%%%%%%%%%%%%%
\subsection{Kolmogorov's proposal: incompressible strings}
%%%%%%%%%%%%%%%%%%%%%%%%%%%%%%%%%%%
%
We now assume that $\+O=\{0,1\}^*$, i.e., we restrict to words.
%
%------------------------------------------------------------------
\subsubsection{Incompressibility with Kolmogorov complexity}
%------------------------------------------------------------------
Though much work had been devoted to get
{\em a mathematical theory of random objects},
notably by von Mises (\cite{mises19,mises39}, 1919-1939),
none was satisfactory up to the 60's when Kolmogorov
based such a theory on Kolmogorov complexity,
hence on computability theory.
\\
The theory was, in fact, independently\footnote{
For a detailed analysis of {\em who did what, and when},
see Li \& Vitanyi's book \cite{livitanyi}, p.89--92.} developed by
Gregory J. Chaitin (b. 1947), 
\cite{chaitin66,chaitin69}
who submitted both papers in 1965.
\medskip\\
The basic idea is as follows:
\medskip
\\{\em \textbullet~ larger is the Kolmogorov complexity of a text,
more random is the text,
\\\textbullet~  larger is its information content,
and more compressed is the text.}
\medskip
\\
Thus, a theory for measuring the information content
is also a theory of randomness.
\medskip\\
Recall that there exists $c$ such that for all $x\in\{0,1\}^*$,
$K(x)\leq |x|+c$ (Proposition \ref{p:bound}).
The reason being that there is a ``stupid" program of length
about $|x|$ which computes the word $x$ by
telling what are the successive letters of $x$.
The intuition of incompressibility is as follows:
$x$ is incompressible if there no shorter way to get $x$.
\\
Of course, we are not going to define absolute randomness
for words. But a measure of randomness telling
{\em how far from $|x|$ is $K(x)$.}
\begin{definition}[Measure of incompressibility]$\\ $
A word $x$ is $c$-incompressible if $K(x)\geq|x|-c$.
\end{definition}
\noindent
It is rather intuitive that most things are random.
The next Proposition formalizes this idea.
\begin{proposition}\label{p:incompressible}
For any $n$, the proportion of $c$-incompressible strings
of length $n$ is $\geq 1-2^{-c}$.
\end{proposition}
\begin{quote}
\begin{proof}
At most $2^{n-c}-1$ programs of length $<n-c$
and $2^n$ strings of length $n$.
\end{proof}
\end{quote}
%
%
%------------------------------------------------------------------
\subsubsection{Incompressibility with length conditional
Kolmogorov complexity}
%------------------------------------------------------------------
%
We observed in \S\ref{sss:entropy} that the entropy of a word
of the form $000...0$ is null.
i.e., entropy did not considered the information conveyed by
the length.
\\
Here, with incompressibility based on Kolmogorov complexity,
we can also ignore the information content conveyed by the length
by considering {\em incompressibility based on length conditional
Kolmogorov complexity}.
\begin{definition}[Measure of length conditional incompressibility]
A word $x$ is length conditional $c$-incompressible if
$K(x\mid |x|)\geq|x|-c$.
\end{definition}
\noindent
The same simple counting argument yields the following Proposition.
\begin{proposition}
For all $n$, the proportion of length conditional $c$-incompressible
strings of length $n$ is $\geq 1-2^{-c}$.
\end{proposition}
\noindent
A priori length conditional incompressibility is stronger
than mere incompressibility.
However, the two notions of incompressibility are about the
same \ldots up to a constant.
\begin{proposition}
There exists $d$ such that, for all 
$c\in\mathbb{N}$ and $x\in\{0,1\}^*$
\medskip\\
1. $x$ is length conditional $c$-incompressible
$\Rightarrow$ $x$ is $(c+d)$-incompressible
\medskip\\
2. $x$ is $c$-incompressible $\Rightarrow$ $x$ is length conditional
$(2c+d)$-incompressible.
\end{proposition}
\begin{quote}
\begin{proof}
1 is trivial.
For 2, first observe that there exists $e$ such that, for all $x$,
$$
(*)\ \ \ \ K(x) \leq K(x\mid |x|) + 2 K(|x| - K(x\mid |x|)) + d
$$
In fact, if $K=K_\varphi$ and $K(\ \mid\ )=K_{\psi(\ \mid\ )}$,
consider $p,q$ such that
\medskip\\
$\begin{array}{rclcrcl}
|x| - K(x\mid |x|) &=& \varphi(p)
&& \psi(q\mid|x|) &=& x
\\
K(|x| - K(x\mid |x|)) &=& |p|
&& K(x\mid |x|) &=& |q|
\end{array}$
\medskip\\
With $p$ and $q$, hence with $\langle p,q\rangle$
(cf. Proposition \ref{p:code}), one can successively get
$\left\{\begin{array}{ll}
|x| - K(x\mid |x|) & \mbox{this is $\varphi(p)$}
\\
K(x\mid |x|) & \mbox{this is $q$}
\\
|x| & \mbox{just sum the above quantities}
\\
x & \mbox{this is $\psi(q\mid|x|)$}
\end{array}\right.$
\\
Thus, $K(x)\leq |\langle p,q\rangle| +O(1)$.
Applying Proposition~\ref{p:code}, we get (*).
\\
Using  $K^{\mathbb{N}}\leq\log+c_1$
and $K^{\words}(x)\geq|x|-c$ (cf., Proposition~\ref{p:bound}),
(*) yields
$$
|x|- K(x\mid |x|) \leq 2\log(|x| - K(x\mid |x|)) + 2c_1+c+d
$$
Finally, observe that $z\leq 2\log z +k$ insures $z\leq\max(8,2k)$.
\end{proof}
\end{quote}
%
%
%%%%%%%%%%%%%%%%%%%%%%%%%%%%%%%%%%%
\subsection{Incompressibility is randomness: Martin-L\"{o}f's argument}
\label{ss:tests}
%%%%%%%%%%%%%%%%%%%%%%%%%%%%%%%%%%%
%
Now, if incompressibility is clearly a necessary condition for
randomness, how do we argue that it is a sufficient condition?
Contraposing the wanted implication, let us see that if a word
fails some statistical test then it is not incompressible.
We consider some spectacular failures of statistical tests.
\begin{example}\label{ex:ex1}
\mbox{}\\
\mbox{}\\ 1. {\em [Constant half length prefix]}
For all $n$ large enough, a string $0^nu$
with $|u|=n$ cannot be $c$-incompressible.
\medskip\\
2. {\em [Palindromes]}
Large enough palindromes cannot be
$c$-incompressible.
\medskip\\
3. {\em [$0$ and $1$ not equidistributed]}
For all $0<\alpha<1$, for all $n$ large enough,
a string of length $n$ which has $\leq \alpha\frac{n}{2}$ zeros
cannot be $c$-incompressible.
\end{example}
\begin{quote}
\begin{proof}
1. Let $c'$ be such that $K(x)\leq|x|+c'$.
Observe that there exists $c''$ such that $K(0^nu) \leq K(u)+c''$
hence
$$
K(0^nu)\leq n+c'+c'' \leq \frac{1}{2}|0^nu|+c'+c''
$$
So that $K(0^nu)\geq|0^nu|-c$ is impossible for $n$ large enough.
\medskip\\
2. Same argument:
There exists $c''$ such that, for any palindrome $x$,
$$
K(x)\leq \frac{1}{2}|x| + c''
$$
\medskip\\
3. The proof follows the classical argument to get the law
of large numbers (cf. Feller's book \cite{feller}).
Let us do it for $\alpha=\frac{2}{3}$, so that
$\frac{\alpha}{2}=\frac{1}{3}$.
\medskip\\
Let $A_n$ be the set of strings of length $n$ with
$\leq \frac{n}{3}$ zeros.
We estimate the number $N$ of elements of $A_n$.
$$
N=\sum_{i=0}^{i=\frac{n}{3}}
\left(\begin{array}{c} n \\i \end{array}\right)
\leq
(\frac{n}{3}+1)\
\left(\begin{array}{c} n \\ \frac{n}{3} \end{array}\right)
=\ (\frac{n}{3}+1)\ \frac{n!}{\frac{n}{3}!\ \frac{2n}{3}!}
$$
Use inequality $1 \leq e^{\frac{1}{12n}} \leq 1.1$
and Stirling's formula (1730),
$$
\sqrt{2n\pi}\ {\left(\frac{n}{e}\right)}^n\ e^{\frac{1}{12n+1}}
< n!
< \sqrt{2n\pi}\ {\left(\frac{n}{e}\right)}^n\ e^{\frac{1}{12n}}
$$
Observe that $1.1\ (\frac{n}{3}+1) < n$ for $n\geq2$. Therefore,
$$
N < n \frac{\sqrt{2n\pi}\ {\left(\frac{n}{e}\right)}^n}
{\sqrt{2\frac{n}{3}\pi}\
{\left(\frac{\frac{n}{3}}{e}\right)}^{\frac{n}{3}}\
\sqrt{2\frac{2n}{3}\pi}\
{\left(\frac{\frac{2n}{3}}{e}\right)}^{\frac{2n}{3}}}\
= \frac{3}{2}\
  \sqrt{\frac{n}{\pi}}\ {\left(\frac{3}{\sqrt[3]{4}}\right)}^n
$$
Using Proposition \ref{p:rank}, for any element of $A_n$, we have
$$
K(x\mid n)
\leq \log(N) +d
\leq n\log\left(\frac{3}{\sqrt[3]{4}}\right) + \frac{\log n}{2} + d
$$
Since $\frac{27}{4} < 8$, we have $\frac{3}{\sqrt[3]{4}} < 2$
and $\log\left(\frac{3}{\sqrt[3]{4}}\right) < 1$.
Hence,
$n-c \leq n\log\left(\frac{3}{\sqrt[3]{4}}\right)+\frac{\log n}{2}+d$
is impossible for $n$ large enough.
\\
So that $x$ cannot be $c$-incompressible.
\end{proof}
\end{quote}
Let us give a common framework to the three above examples
so as to get some flavor of what can be a statistical test.
To do this, we follow the above proofs of compressibility.
\begin{example}\label{ex:ex2}
\mbox{}\\
\mbox{}\\ 1. {\em [Constant left half length prefix]}\\
Set $V_m = \mbox{ all strings with $m$ zeros ahead}$.
The sequence $V_0,V_1,...$ is decreasing.
The number of strings of length $n$ in $V_m$ is $0$ if $m>n$
and $2^{n-m}$ if $m\leq n$.
Thus, the proportion
$\frac{\sharp\{x \mid |x|=n\ \wedge\ x\in V_m\}}{2^n}$
of length $n$ words which are in $V_m$ is $2^{-m}$.
\medskip\\
2. {\em [Palindromes]}
Put in $V_m$ all strings which have equal length $m$
prefix and suffix.
The sequence $V_0,V_1,...$ is decreasing.
The number of strings of length $n$ in $V_m$ is
$0$ if $m>\frac{n}{2}$
and $2^{n-2m}$ if $m\leq \frac{n}{2}$.
Thus, the proportion of length $n$ words which are in $V_m$
is $2^{-2m}$.
\medskip\\
3. {\em [$0$ and $1$ not equidistributed]}
Put in $V^\alpha_m =$ all strings $x$ such that
the number of zeros is
$\leq (\alpha+(1-\alpha)2^{-m})\frac{|x|}{2}$.
The sequence $V_0,V_1,...$ is decreasing.
A computation analogous to that done in the proof of the law
of large numbers shows that the proportion of length $n$ words
which are in $V_m$ is $\leq 2^{-\gamma m}$ for some $\gamma>0$
(independent of $m$).
\end{example}
\noindent
Now, what about other statistical tests?
But what is a statistical test?
A convincing formalization has been developed by Martin-L\"of.
The intuition is that illustrated in Example \ref{ex:ex2}
augmented of the following feature:
each $V_m$ is computably enumerable and so is the relation
$\{(m,x) \mid x\in V_m\}$.
A feature which is analogous to the partial computability
assumption in the definition of Kolmogorov complexity.
\begin{definition}\label{def:test}
[Abstract notion of statistical test, Martin-L\"of, 1964]
A statistical test is a family of nested critical sets
$$
\{0,1\}^*\supseteq V_0\supseteq V_1\supseteq V_2\supseteq
...\supseteq V_m\supseteq...
$$
such that $\{(m,x) \mid x\in V_m\}$ is computably enumerable
and the proportion
$\frac{\sharp\{x \mid |x|=n\ \wedge\ x\in V_m\}}{2^n}$
of length $n$ words which are in $V_m$ is $\leq 2^{-m}$.
\medskip\\
Intuition. The bound $2^{-m}$ is just a normalization.
Any bound $b(n)$ such that 
$b:\mathbb{N}\to\mathbb{Q}$ which is computable, 
decreasing and with limit $0$ could replace $2^{-m}$. 
\\
The significance of $x\in V_m$ is that the hypothesis
{\em $x$ is random} is rejected with significance level $2^{-m}$.
\end{definition}
\begin{remark}
Instead of sets $V_m$ one can consider a function
$\delta:\{0,1\}^*\to\mathbb{N}$ such that
$\frac{\sharp\{x \mid |x|=n\ \wedge\ \delta(x)\geq m\}}{2^n}\leq 2^{-m}$
and $\delta$ is computable from below, i.e.,
$\{(m,x) \mid \delta(x)\geq m\}$ is recursively enumerable.
\end{remark}
We have just argued on some examples that all statistical tests
from practice are of the form stated by Definition \ref{def:test}.
Now comes Martin-L\"of fundamental result about statistical tests
which is in the vein of the invariance theorem.
\begin{theorem} [Martin-L\"of, 1965]\label{thm:test}
Up to a constant shift, there exists a largest statistical test
$(U_m)_{m\in\mathbb{N}}$
$$
\forall (V_m)_{m\in\mathbb{N}}\quad \exists c\quad \forall m\quad
V_{m+c} \subseteq U_m
$$
In terms of functions, up to an additive constant,
there exists a largest statistical test $\Delta$
$$
\forall \delta\quad \exists c\quad \forall x\quad
\delta(x)<\Delta(x)+c
$$
\end{theorem}
\begin{quote}
\begin{proof}
Consider $\Delta(x)=|x|-K(x\mid|x|)-1$.
\\
\fbox{\em $\Delta$ is a test.}
Clearly, $\{(m,x) \mid \Delta(x)\geq m\}$ is computably enumerable.
\\
$\Delta(x)\geq m$ means $K(x\mid|x|)\leq |x|-m-1$.
So no more elements in $\{x \mid \Delta(x)\geq m\ \wedge\ |x|=n\}$
than programs of length $\leq n-m-1$, which is $2^{n-m}-1$.
\\
\fbox{\em $\Delta$ is largest.}
$x$ is determined by its rank in the set
$V_{\delta(x)}=\{z \mid \delta(z)\geq \delta(x)\ \wedge\ |z|=|x|\}$.
Since this set has $\leq 2^{n-\delta(x)}$ elements,
the rank of $x$ has a binary representation of length
$\leq |x|-\delta(x)$.
Add useless zeros ahead to get a word $p$ with length
$|x|-\delta(x)$.
\\
With $p$ we get $|x|-\delta(x)$.
With $|x|-\delta(x)$ and $|x|$ we get $\delta(x)$ and construct
$V_{\delta(x)}$. With $p$ we get the rank of $x$ in this set,
hence we get $x$.
Thus,\\
$K(x\mid|x|) \leq |x|-\delta(x) +c$,
i.e., $\delta(x)<\Delta(x)+c$.
\end{proof}
\end{quote}
The importance of the previous result is the following corollary
which insures that, for words,
incompressibility implies (hence is equivalent to) randomness.
\begin{corollary}[Martin-L\"of, 1965]
Incompressibility passes all statistical tests.
I.e., for all $c$, for all statistical test $(V_m)_m$,
there exists $d$ such that
$$
\forall x\ (x\mbox{ is $c$-incompressible }
\Rightarrow\ x\notin V_{c+d})
$$
\end{corollary}
\begin{quote}
\begin{proof}
Let $x$ be length conditional $c$-incompressible.
This means that $K(x\mid |x|)\geq |x|-c$.
Hence $\Delta(x)=|x|-K(x\mid|x|)-1\leq c-1$, which means that
$x\notin U_c$.
\\
Let now $(V_m)_m$ be a statistical test.
Then there is some $d$ such that $V_{m+d}\subseteq U_m$
Therefore $x\notin V_{c+d}$.
\end{proof}
\end{quote}
\begin{remark}
Observe that incompressibility is a {\em bottom-up} notion:
we look at the value of $K(x)$ (or that of $K(x \mid |x|)$).
\\
On the opposite, passing statistical tests is a {\em top-down}
notion.
To pass all statistical tests amounts to an inclusion
in an intersection: namely, an inclusion in
$$
\bigcap_{(V_m)_m}\ \bigcup_c\ V_{m+c}
$$
\end{remark}
%
%
%
%%%%%%%%%%%%%%%%%%%%%%%%%%%%%%%%%%%
\subsection{Shortest programs are random finite strings}
\label{ss:bestprograms}
%%%%%%%%%%%%%%%%%%%%%%%%%%%%%%%%%%%
%
Observe that optimal programs to compute any object
are examples of random strings.
More precisely, the following result holds.
\begin{proposition}
Let $\+O$ be an elementary set (cf. Definition \ref{def:elementary})
and $U:\words\to \words $, $V:\words\to\+O$ be some fixed optimal
functions.
There exists a constant $c$ such that, for all $a\in\+O$,
for all $p\in\words$,
if $V(p)=a$ and $K_V(a)=|p|$ then $K_U(p)\geq |p|-c$.
In other words, for any $a\in\+O$,
if $p$ is a shortest program which outputs $a$ then $p$ is $c$-random.
\end{proposition}
\begin{proof}
Consider the function $V\circ U: \words\to\+O$.
Using the invariance theorem, let $c$ be such that
$K_V \leq K_{V\circ U} +c$.
Then, for every $q\in\words$,
\begin{eqnarray*}
U(q)=p &\Rightarrow& V\circ U(q)=a\\
&\Rightarrow& |q| \geq K_{V\circ U}(a) \geq K_V(a)-c = |p|-c
\end{eqnarray*}
Which proves that $K_U(p)\geq |p|-c$.
\end{proof}
%
%
%
%%%%%%%%%%%%%%%%%%%%%%%%%%%%%%%%%%%
\subsection{Random finite strings and
            lower bounds for computational complexity}
\label{ss:lowerbounds}
%%%%%%%%%%%%%%%%%%%%%%%%%%%%%%%%%%%
%
Random finite strings (or rather $c$-incompressible strings)
have been extensively used to prove lower bounds for computational
complexity, cf. the pioneering paper \cite{paul79}
by Wolfgang Paul, 1979,
(see also an account of the proof in our survey paper
\cite{ferbusgrigoBullEATCS2001})
and the work by Li \& Vitanyi, \cite{livitanyi}.
The key idea is that a random string can be used as a worst possible
input.
%
%
%
%
%%%%%%%%%%%%%%%%%%%%%%%%%%%%%%%%%%%
%%%%%%%%%%%%%%%%%%%%%%%%%%%%%%%%%%%
%%%%%%%%%%%%%%%%%%%%%%%%%%%%%%%%%%%
\section{Formalization of randomness: infinite objects}
%%%%%%%%%%%%%%%%%%%%%%%%%%%%%%%%%%%
%%%%%%%%%%%%%%%%%%%%%%%%%%%%%%%%%%%
%%%%%%%%%%%%%%%%%%%%%%%%%%%%%%%%%%%
%
We shall stick to infinite sequences of zeros and ones:
$\{0,1\}^\mathbb{N}$.
%
%
%%%%%%%%%%%%%%%%%%%%%%%%%%%%%%%%%%%
\subsection{Martin-L\"of top-down approach
with topology and computability}
%%%%%%%%%%%%%%%%%%%%%%%%%%%%%%%%%%%
%
%------------------------------------------------------------------
\subsubsection{The naive idea badly fails}
%------------------------------------------------------------------
%
The naive idea of a random element of $\cantor$ is that of a
sequence $\alpha$ which is in no set of measure $0$.
Alas, $\alpha $ is always in the singleton set $\{\alpha\}$
which has measure $0$ !
%
%------------------------------------------------------------------
\subsubsection{Martin-L\"of's solution: effectivize}
\label{sss:MLrandom}
%------------------------------------------------------------------
%
Martin-L\"of's solution to the above problem is to effectivize,
i.e., to consider the sole effective measure zero sets.
\\
This approach is, in fact, an extension to infinite sequences
of the one Martin-L\"of developed for finite objects,
cf. \S\ref{ss:tests}.
\medskip\\
Let us develop a series of observations which leads to
Martin-L\"of's precise solution, i.e., what does mean effective
for measure $0$ sets.
\\
To prove a probability law amounts to prove that
a certain set $X$ of sequences has probability one.
To do this, one has to prove that the complement set
$Y=\cantor \setminus X$ has probability zero.
Now, in order to prove that $Y \subseteq \cantor$ has
probability zero, basic measure theory tells us that
one has to include $Y$ in open sets with
arbitrarily small probability.
I.e., for each $n\in \mathbb{N}$ one must find an open set
$U_{n}\supseteq Y$
which has probability $\leq \frac{1}{2^n}$.
\\
If things were on the real line ${\mathbb{R}}$
we would say that $U_{n}$ is a countable union of
intervals with rational endpoints.
\\
Here, in $\cantor$, $U_{n}$ is a
countable union of sets of the form
$u\cantor$ where $u$ is a finite binary
string and $u\cantor$ is the set of infinite sequences
which extend~$u$.
\\
In order to prove that $Y$ has probability zero,
for each $n\in \mathbb{N}$ one must find a family
$(u_{n,m})_{m\in\mathbb{N}}$ such that
$Y\subseteq \bigcup_{m} u_{n,m}\cantor$
and $Proba(\bigcup_{m} u_{n,m}\cantor)\leq \frac{1}{2^n}$
for each $n\in \mathbb{N}$.
\\
Now, Martin-L\"of makes a crucial observation:
mathematical probability laws which we consider
necessarily have some effective character.
And this effectiveness should reflect in the proof
as follows:
{\em the doubly indexed sequence
$(u_{n,m})_{{n,m\in\mathbb{N}}}$ is computable.}
\medskip\\
Thus, the set $\bigcup_{m} u_{n,m}\cantor$ is a
{\em computably enumerable open set} and
$\bigcap_{n} \bigcup_{m} u_{n,m}\cantor$
is a countable intersection of a
{\em computably enumerable family of open sets}.
\medskip\\
Now comes the essential theorem, which is completely analogous
to Theorem~\ref{thm:test}.
\begin{definition}[Martin-L\"of, \cite{martinlof66}, 1966]
\label{def:null}
A constructively null $G_\delta$ set is any set of the form
$$
\bigcap_{n} \bigcup_{m} u_{n,m}\cantor
$$
where $Proba(\bigcup_{m} u_{n,m}\cantor)\leq \frac{1}{2^n}$
(which implies that the intersection set has probability zero)
and the sequence $u_{n,m}$ is computably enumerable.
\end{definition}
\begin{theorem}[Martin-L\"of, \cite{martinlof66}, 1966]
There exist a largest constructively null $G_\delta$ set
\end{theorem}
\noindent
Let us insist that the theorem says {\em largest}, up to nothing, 
really largest relative to set inclusion.
\begin{definition}[Martin-L\"of, \cite{martinlof66}, 1966] 
A sequence $\alpha\in\cantor$ is Martin-L\"of random
if it belongs to no constructively null $G_\delta$ set
(i.e., if it does not belongs to the largest one).
\end{definition}
\noindent
In particular, the family of random sequences, being the complement
of a constructively null $G_\delta$ set, has probability $1$.
And the observation above Definition~\ref{def:null} insures that
Martin-L\"of random sequences satisfy all usual probabilities laws.
Notice that {\em the last statement can be seen as an improvement of
all usual probabilities laws: not only such laws are true with
probability $1$ but they are true for all sequences in the measure
$1$ set of Martin-L\"of random sequences.}
%
%
%
%%%%%%%%%%%%%%%%%%%%%%%%%%%%%%%%%%%
\subsection{The bottom-up approach}
%%%%%%%%%%%%%%%%%%%%%%%%%%%%%%%%%%%
%
%------------------------------------------------------------------
\subsubsection{The naive idea badly fails}
%------------------------------------------------------------------
Another natural naive idea to get randomness for sequences is to
extend randomness from finite objects to infinite ones.
The obvious proposal is to consider sequences
$\alpha\in\cantor$ such that, for some $c$,
\begin{equation}
\forall n\quad K(\alpha\segment n)\geq n-c
\end{equation}
However, Martin-L\"of proved that there is no such sequence.
\begin{theorem}
[Large oscillations (Martin-L\"of, \cite{martinlof71}, 1971)]
\label{thm:osc}
If $f:\mathbb{N}\to\mathbb{N}$ is computable
and $\sum_{n\in\mathbb{N}}2^{-f(n)}=+\infty$
then, for every $\alpha\in\cantor$, there are infinitely many $k$
such that $K(\alpha\segment k) \leq k-f(k) -O(1)$.
\end{theorem}
\begin{quote}
\begin{proof}
Let us do the proof in the case $f(n)=\log n$ which is quite limpid
(recall that the harmonic series $\frac{1}{n}=2^{-\log n}$
has infinite sum).
\\
Let $k$ be any integer.
The word $\alpha\segment k$ prefixed with $1$
is the binary representation of an integer $n$
(we put $1$ ahead of $\alpha\segment k$ in order to avoid
a first block of non significative zeros).
We claim that $\alpha\segment n$ can be recovered from
$\alpha\segment [k+1,n]$ only. In fact,
\begin{itemize}
\item
$n-k$ is the length of $\alpha\segment [k+1,n]$,
\item
$k=\lfloor\log n\rfloor+1
      =\lfloor\log(n-k)\rfloor+1+\varepsilon$\
(where $\varepsilon\in\{0,1\}$)
is known from $n-k$ and $\varepsilon$,
\item
$n=(n-k)+k$.
\item
$\alpha\segment k$ is the binary representation of $n$.
\end{itemize}
The above analysis describes a computable map
$f:\{0,1\}^*\times\{0,1\}\to\{0,1\}^*$
such that
$\alpha\segment n=f(\alpha\segment [k+1,n],\varepsilon)$.
Applying Proposition~\ref{p:bound}, point 3, we get
$$
K(\alpha\segment n) \leq K(\alpha\segment [k+1,n]) +O(1)
\leq n-k+O(1) = n-\log(n)+O(1)
$$
\end{proof}
\end{quote}
%
%
%
%
%------------------------------------------------------------------
\subsubsection{Miller \& Yu's theorem}
%------------------------------------------------------------------
%
It took about forty years to get a characterization of randomness
via Kolmogorov complexity which completes Theorem \ref{thm:osc}
in a very pleasant and natural way.
\begin{theorem}[Miller \& Yu, \cite{milleryu}, 2008]
\label{thm:milleryu}
The following conditions are equivalent:
\begin{enumerate}
\item[i.\quad]
The sequence $\alpha\in\cantor$ is Martin-L\"of random
\item[ii.\quad]
$\exists c\quad\forall k\quad
K(\alpha\segment k) \geq k - f(k) -c$
for every total computable function
$f:\mathbb{N}\to\mathbb{N}$ satisfying
$\sum_{n\in\mathbb{N}}2^{-f(n)}<+\infty$
\item[iii.\quad]
$\exists c\quad\forall k\quad
K(\alpha\segment k) \geq k - H(k) -c$
\end{enumerate}
Moreover, there exists a particular total computable 
function $g:\mathbb{N}\to\mathbb{N}$ satisfying 
$\sum_{n\in\mathbb{N}}2^{-g(n)}<+\infty$ such that
one can add a fourth equivalent condition:
\begin{enumerate}
\item[iv.\quad]
$\exists c\quad\forall k\quad
K(\alpha\segment k) \geq k - g(k) -c$
\end{enumerate}
\end{theorem}
\noindent
Recently, an elementary proof of this theorem
was given by Bienvenu, Merkle \& Shen in
\cite{bms08}, 2008.
Equivalence $i\Leftrightarrow iii$ is due to
G\'acs, \cite{gacs80}, 1980.
%
%
%
%
%------------------------------------------------------------------
\subsubsection{Variants of Kolmogorov complexity and randomness}
\label{sss:SKmH}
%------------------------------------------------------------------
Bottom-up characterization of random sequences have been obtained
using Levin monotone complexity,
Schnorr process complexity and prefix complexity
(cf. \S\ref{ss:monotone}, \S\ref{ss:Schnorr} and \S\ref{ss:self}).
\begin{theorem}\label{thm:SKmH}
The following conditions are equivalent:
\begin{enumerate}
\item[i.]\quad
The sequence $\alpha\in\cantor$ is Martin-L\"of random
\item[ii.]\quad
$\exists c\quad\forall k\quad
|K^{mon}(\alpha\segment k) -k| \leq c$
\item[iii.]\quad
$\exists c\quad\forall k\quad |S(\alpha\segment k) -k| \leq c$
\item[iv.]\quad
$\exists c\quad\forall k\quad H(\alpha\segment k) \geq k - c$
\end{enumerate}
\end{theorem}
\noindent
Equivalence $i\Leftrightarrow ii$ is due to Levin
(\cite{zvonkin-levin}, 1970).
Equivalence $i\Leftrightarrow iii$ is due to
Schnorr (\cite{schnorr71}, 1971).
Equivalence $i\Leftrightarrow iv$ is due to Schnorr and
Chaitin (\cite{chaitin75}, 1975).
%
%
%
%
%%%%%%%%%%%%%%%%%%%%%%%%%%%%%%%%%%%
\subsection{Randomness: a robust mathematical notion}
%%%%%%%%%%%%%%%%%%%%%%%%%%%%%%%%%%%
%
Besides the top-down definition of Martin-L\"of randomness,
we mentioned above diverse bottom-up characterizations
via properties of the initial segments
with respect to variants of Kolmogorov complexity.
There are other top-down and bottom-up characterizations,
we mention two of them in this \S.
\noindent\\
This variety of characterizations shows that Martin-L\"of randomness
is a robust mathematical notion.
%
%------------------------------------------------------------------
\subsubsection{Randomness and martingales}
%------------------------------------------------------------------
%
Recall that a martingale is a function $d:\{0,1\}^*\to\mathbb{R}^+$
such that
$$
\forall u\quad d(u)=\frac{d(u0)+d(u1)}{2}
$$
The intuition is that a player tries to predict the bits of a
sequence $\alpha\in\cantor$ and bets some amount of money
on the values of these bits.
If his guess is correct he doubles his stake, else he looses it.
Starting with a positive capital $d(\varepsilon)$
(where $\varepsilon$ is the empty word), $d(\alpha\segment k)$
is his capital after the $k$ first bits of $\alpha$ have been
revealed.
\\
The martingale $d$ wins on $\alpha\in\cantor$ if the capital
of the player tends to $+\infty$.
\\
The martingale $d$ is computably approximable from below if
the left cut of $d(u)$ is computably enumerable, uniformly in $u$
(i.e., $\{(u,q)\in\{0,1\}^*\times\mathbb{Q} \mid q\leq d(u)\}$ is c.e.).
\begin{theorem}[Schnorr, \cite{schnorr71a}, 1971]
A sequence $\alpha\in\cantor$ is Martin-L\"of random
if and only if
no martingale computably approximable from below wins on $\alpha$.
\end{theorem}
%
%
%------------------------------------------------------------------
\subsubsection{Randomness and compressors}
%------------------------------------------------------------------
%
Recently, Bienvenu \& Merkle obtained quite remarkable
characterizations of random sequences in the vein of
Theorems~\ref{thm:SKmH} and \ref{thm:milleryu}
involving {\em computable} upper bounds of $K$ and $H$.
\begin{definition}
A compressor is any partial computable $\Gamma:\{0,1\}^*\to\{0,1\}^*$
which is one-to-one and has computable domain.
A compressor is said to be prefix-free if its range is prefix-free.
\end{definition}
\begin{proposition}
\mbox{}\\
\mbox{}\\	
1. If $\Gamma$ is a compressor (resp. a prefix-free compressor)
then
\begin{center}
$\exists c\quad\forall x\in\{0,1\}^*\quad K(x) \leq |\Gamma(x)|+c$
\qquad(resp. $H(x) \leq |\Gamma(x)|+c$)
\end{center}
2. For any computable upper bound $F$ of $K$ (resp. of $H$)
there exists a compressor (resp. a prefix-free compressor) $\Gamma$
such that
$$
\exists c\quad\forall x\in\{0,1\}^*\quad |\Gamma(x)| \leq F(x)+c
$$
\end{proposition}
\noindent
Now comes the surprising characterizations of randomness
in terms of {\em computable functions}.
\begin{theorem}[Bienvenu \& Merkle, \cite{bm07}, 2007]
The following conditions are equivalent:
\begin{enumerate}
\item[i.]\quad
The sequence $\alpha\in\cantor$ is Martin-L\"of random
\item[ii.]\quad
For all prefix-free compressor $\Gamma:\{0,1\}^*\to\{0,1\}^*$,
$$
\exists c\quad \forall k\quad |\Gamma(\alpha\segment k)| \geq k - c
$$
\item[iii.]\quad
For all compressor $\Gamma$,\quad
$\exists c\quad \forall k\quad
|\Gamma(\alpha\segment k)| \geq k - H(k) - c$
\end{enumerate}
Moreover, there exists a particular prefix-free compressor
$\Gamma^*$ and a particular compressor $\Gamma^\#$ such that
one can add two more equivalent conditions:
\begin{enumerate}
\item[iv.]\quad
$\exists c\quad\forall k\quad
|\Gamma^*(\alpha\segment k)| \geq k - c$
\item[v.]\quad
$\exists c\quad\forall k\quad
|\Gamma^\#(\alpha\segment k)| 
\geq k - |\Gamma^*(\alpha\segment k)| - c$
\end{enumerate}
\end{theorem}
%
%
%
%%%%%%%%%%%%%%%%%%%%%%%%%%%%%%%%%%%
\subsection{Randomness: a fragile property}
%%%%%%%%%%%%%%%%%%%%%%%%%%%%%%%%%%%
%
Though the notion of Martin-L\"of randomness is robust,
with a lot of equivalent definitions,
as a property, it is quite fragile.
\\
In fact, random sequences loose their random character under
very simple computable transformation.
For instance, even if $a_0a_1a_2...$ is random, the sequence
$0 a_0 0 a_1 0 a_2 0...$ IS NOT random since it fails the
following Martin-L\"of test:
$$
\bigcap_{n\in\mathbb{N}}\{\alpha \mid \forall i<n\ \alpha(2i+1)=0\}
$$
Indeed, $\{\alpha \mid \forall i<n\ \alpha(2i+1)=0\}$
has probability $2^{-n}$ and is an open subset of $\cantor$.
%
%
%
%%%%%%%%%%%%%%%%%%%%%%%%%%%%%%%%%%%
\subsection{Randomness is not chaos}
%%%%%%%%%%%%%%%%%%%%%%%%%%%%%%%%%%%
%
In a series of  papers \cite{joan1,joan2,joan3}, 1993-1996,
Joan Rand Moschovakis introduced a 
very convincing notion of chaotic sequence $\alpha\in\cantor$.
It turns out that the set of such sequences has measure zero
and is disjoint from Martin-L\"of random sequences.
\\
This stresses that {\em randomness is not chaos}.
As mentioned in \S\ref{sss:MLrandom},
random sequences obey laws, those of probability theory.
%
%
%
%%%%%%%%%%%%%%%%%%%%%%%%%%%%%%%%%%%
\subsection{Oracular randomness}
%%%%%%%%%%%%%%%%%%%%%%%%%%%%%%%%%%%
%
%------------------------------------------------------------------
\subsubsection{Relativization}
\label{sss:nrandom}
%------------------------------------------------------------------
%
Replacing ``computable" by ``computable in some oracle",
all the above theory relativizes in an obvious way,
using oracular Kolmogorov complexity and the oracular variants.
\\
In particular, when the oracle is the halting problem, i.e.
the computably enumerable set $\emptyset'$,
the obtained randomness is called $2$-randomness.
\\
When the oracle is the halting problem of partial
$\emptyset'$-computable functions, i.e.
the computably enumerable set $\emptyset''$,
the obtained randomness is called $3$-randomness.
And so on.
\\
Of course, $2$-randomness implies randomness
(which is also called $1$-randomness)
and $3$-randomness implies $2$-randomness.
And so on.
%
%------------------------------------------------------------------
\subsubsection{Kolmogorov randomness and $\emptyset'$}
\label{sss:nies}
%------------------------------------------------------------------
A natural question following Theorem \ref{thm:osc} is to look at
the so-called {\em Kolmogorov random sequences} which satisfy
$K(\alpha\segment k)\geq k-O(1)$
for infinitely many $k$'s.
This question got a very surprising answer involving $2$-randomness.
\begin{theorem}[Nies, Stephan \& Terwijn, \cite{nies},  2005]
\label{thm:nies}
Let $\alpha\in\cantor$.
There are infinitely many $k$
such that, for a fixed $c$, $K(\alpha\segment k) \geq k-c$
(i.e., $\alpha$ is Kolmogorov random)
if and only if $\alpha$ is $2$-random.
\end{theorem}
%
%
%%%%%%%%%%%%%%%%%%%%%%%%%%%%%%%%%%%
\subsection{Randomness: a new foundation for probability theory?}
%%%%%%%%%%%%%%%%%%%%%%%%%%%%%%%%%%%
%
Now that there is a sound mathematical notion of randomness,
is it possible/reasonable to use it as a new foundation for
probability theory?
\\
Kolmogorov has been ambiguous on this question.
In his first paper on the subject,
see  p. 35--36 of \cite{kolmo65}, 1965,
he briefly evoked that possibility :
\begin{quote}\em
\dots to consider the use of the
[Algorithmic Information Theory] constructions
in providing a new basis for Probability Theory.
\end{quote}
However, later, see p. 35--36 of \cite{kolmo83}, 1983,
he separated both topics:
\begin{quote}\em
``there is no need whatsoever to change the
established construction of the mathematical
probability theory on the basis of the general theory
of measure.
I am not enclined to attribute the significance of
necessary foundations of probability theory to the
investigations [about Kolmogorov complexity] that
I am now going to survey.
But they are most interesting in themselves.
\end{quote}
though stressing the role of his new theory of random
objects for {\em mathematics as a whole}
in \cite{kolmo83}, p. 39:
\begin{quote}\em
The concepts of information theory as applied
to infinite sequences give rise to very interesting
investigations, which, without being indispensable
as a basis of probability theory, can acquire a
certain value in the investigation of the
algorithmic side of mathematics as a whole.
\end{quote}
%
%%%%%%%%%%%%%%%%%%%%%%%%%%%%%%%%%%%
%%%%%%%%%%%%%%%%%%%%%%%%%%%%%%%%%%%
%%%%%%%%%%%%%%%%%%%%%%%%%%%%%%%%%%%
%%%%%%%%%%%%%%%%%%%%%%%%%%%%%%%%%%%
%%%%%%%%%%%%%%%%%%%%%%%%%%%%%%%%%%%
%

%
%%%%%%%%%%%%%%%%%%%%%%%%%%%%%%%%%%%%%%
%%%%%%%%%%%%%%%%%%%%%%%%%%%%%%%%%%%%%%
%
%
%%%%%%%%%%%%%%%%%%%%%%%%%%%%%%%%%%%%%
%%%%%%%%%%%%%%%%%%%%%%%%%%%%%%%%%%%%%
\end{document}